%% file: paper.tex
\documentclass[lettersize,journal]{IEEEtran}
\usepackage{amsmath,amsfonts}
\usepackage{algorithm}
\usepackage{algorithmicx}
\usepackage{algpseudocode}
\usepackage{array}
\usepackage[caption=false,font=normalsize,labelfont=sf,textfont=sf]{subfig}
\usepackage{textcomp}
\usepackage{stfloats}
\usepackage{amsmath,amssymb}
\usepackage{mathrsfs}
\usepackage{url}
\usepackage{verbatim}
\usepackage{graphicx}
\usepackage{cite}
\usepackage[mathlines]{lineno}
\usepackage{amsthm}
\usepackage{bbm}
\usepackage{bm}
\usepackage{diagbox}
\usepackage{booktabs}

\usepackage[dvipsnames]{xcolor}
\newtheorem{theorem}{Theorem}
\newtheorem{pro}{Proposition}
\newtheorem{lem}{lemma}
\newtheorem{lemma}{Lemma}

\newtheorem{remark}{Remark}

\usepackage{multicol}
\usepackage{setspace}

\def\vu{{\bm{u}}}
\def\vv{{\bm{v}}}

\def\vx{{\bm{x}}}
\def\vy{{\bm{y}}}

\def\mA{{\bm{A}}}

\def\mX{{\bm{X}}}

\input{preamble}

\title{Solving Quadratic Systems with Full-Rank  Matrices Using  Sparse or Generative Priors\thanks{J. Chen is with the Department of Mathematics, The University of Hong Kong; 
M. K. Ng is with the Department of Mathematics, Hong Kong Baptist University; Z. Liu is with the School of Computer Science and Engineering, University of Electronic Science and Technology of China. {\it (Corresponding author: Michael K. Ng.)}}}
\author{Junren Chen,~Michael K. Ng, {\it Senior Member, IEEE},~Zhaoqiang Liu}

\date{}

\begin{document}
	\maketitle

\begin{abstract}
    The problem of recovering a signal $\bm{x}\in \mathbb{R}^n$ from a quadratic system $\{y_i=\bm{x}^\top\bm{A}_i\bm{x},\ i=1,\ldots,m\}$ with full-rank matrices $\bm{A}_i$ frequently arises in applications such as unassigned distance geometry and sub-wavelength imaging. With i.i.d. standard Gaussian matrices $\bm{A}_i$, this paper addresses the high-dimensional case where $m\ll n$ by incorporating prior knowledge of $\bm{x}$. First, we consider a $k$-sparse $\bm{x}$ and introduce the thresholded Wirtinger flow (TWF) algorithm that does not require the sparsity level $k$. TWF comprises two steps: the spectral initialization that identifies a point sufficiently close to $\bm{x}$ (up to a sign flip) when $m=O(k^2\log n)$, and the thresholded gradient descent which, when provided a good initialization, produces a sequence linearly converging to $\bm{x}$  with $m=O(k\log n)$ measurements. Second, we explore the generative prior, assuming that $\bm{x}$ lies in the range of an $L$-Lipschitz continuous generative model with $k$-dimensional inputs in an $\ell_2$-ball of radius $r$. With an estimate correlated with the signal, we develop the projected gradient descent (PGD) algorithm that also comprises two steps: the projected power method that provides an initial vector with $O\big(\sqrt{\frac{k \log L}{m}}\big)$ $\ell_2$-error given $m=O(k\log(Lnr))$ measurements, and the projected gradient descent that refines the $\ell_2$-error to $O(\delta)$ at a geometric rate when $m=O(k\log\frac{Lrn}{\delta^2})$. Experimental results corroborate our theoretical findings and show that: (i) our approach for the sparse case notably outperforms the existing provable algorithm sparse power factorization; (ii) leveraging the generative prior allows for precise image recovery in the MNIST dataset from a small number of quadratic measurements. 
\end{abstract}

\section{Introduction}
Quadratic systems form the basis for modeling a multitude of scientific problems, encompassing scenarios such as the turnpike and beltway problems \cite{Dakic:2000,Huang:Turnpike:2021}, phase retrieval \cite{Fienup:82,Millane:90,Drenth:crystallography:2007}, cryo-electron microscopy (EM) \cite{Huang:OMR_SC:2022}, sub-wavelength imaging \cite{Shechtman:11,Shechtman:12}, and power flow analysis \cite{Eminoglu:2008,Wang2017power}. In this paper, we delve into the quadratic system defined by measurements $\{y_i=\vx^\top\mA_i\vx,\ i=1,\ldots,m\}$, where $\vx\in\mathbb{R}^n$ represents the signal of interest, and $\mA_i$ is a {\it full-rank} measurement matrix. This measurement model arises in a number of applications:
\begin{itemize}
\item \textbf{The turnpike and beltway problems} focus on reconstructing relative positions of a set of points on a line or loop using unlabeled pairwise distances. By discretizing the 1D spatial domain into segments, the point locations are represented as an indicator vector $\vx\in\{0,1\}^n$, where $x_j=1$ signifies the presence of a point at the $j$-th segment. The distance distribution $\vy$ can then be expressed in a quadratic form $y_i=\vx^\top\mA_i\vx$ \cite{Dakic:2000,Huang:Turnpike:2021}.
\item \textbf{Sub-wavelength imaging} involves recovering sub-wavelength features of an object using partially-spatially-incoherent light. Due to the correlation of light field amplitudes, the intensity measurements $\vy$ can be represented in a full-rank quadratic form $y_i=\vx^\top\mA_i\vx$ with respect to the object's transmittance function $\vx$ \cite{Shechtman:11,Shechtman:12}.
\item \textbf{Cryo-EM} aims to reconstruct the 3D density map of a protein from noisy 2D projections with unknown viewing directions. The method of moments has been employed to extract moment-based features from random projections in the frequency domain \cite{KAM1980}. Specifically, the second-order autocorrelation features $\vy$ have been shown to have a quadratic formulation $y_i=\vx^\top\mA_i\vx$ in relation to the density map $\vx$ \cite{Huang:OMR_SC:2022,Mona:UVT:2020,Mona:SUVT:2019}.
\end{itemize}

Dubbed as ``phase retrieval'', the quadratic system $y_i=\bm{x}^\top(\bm{a}_i\bm{a}_i^\top)\bm{x},~i\in [m]$ with rank-1 measurement matrices $\{\bm{a}_i\bm{a}_i^\top\}_i$ has been studied extensively in the literature (e.g., \cite{candes2015phase,Candes:PhaseLift:2013,chen2017solving,chen2023phase,zhang2017nonconvex,Wang2018Trunc}). However, this formulation falls short in modelling the above applications characterized by full-rank measurement matrices $\{\mA_i\}_i$. Prior research on the full-rank quadratic system has primarily been restricted to low-dimensional cases requiring $m>n$ and does not leverage any prior information (e.g., \cite{Tu2016Procrustes,Huang:Quadratic:2020,Huang:Quad:2019,chen2022error}). In this paper, we   focus on the high-dimensional cases that allow for $m\ll n$ by incorporating a sparse or generative prior on $\vx$.\footnote{Throughout this paper, we assume $m=O(n)$.} We investigate i.i.d. standard Gaussian measurements, a widely studied scenario in phase retrieval, which serves as an important theoretical benchmark and sheds light on the understanding of more practical measurement ensembles.

\subsection{Problem Set-up}\label{sec:setup}

We study a high-dimensional quadratic system $$\{y_i=\bm{x}^\top\bm{A}_i\bm{x},~i=1,...,m\},$$ where the signal dimensionality $n$ may exceed or even be far greater than the number of measurements $m$. In the realm of high-dimensional statistical estimation and related signal recovery challenges, it is widely acknowledged that incorporating specific prior knowledge on $\bm{x}$ is essential for rendering the problem well-defined. This understanding holds true for various problems such as compressed sensing, high-dimensional linear regression, and phase retrieval (as exemplified in \cite{foucart2013invitation,chen2023high,negahban2012unified,bickel2009simultaneous,Wang:TAF:2018,raskutti2011minimax,bora2017compressed,cai2016optimal}, among others). In this paper, we explore two distinct types of priors: the classical sparse prior and the increasingly popular generative prior. The problem is rigorously formulated as follows:

\begin{itemize}
    \item \textbf{Sparse Quadratic System.}\footnote{In this term, ``sparse'' describes the property of the signal $\bm{x}$ (rather than matrix $\bm{A}_i$); the term ``generative quadratic system'' below is similar.}
    Assuming that $\bm{A}_1, \ldots, \bm{A}_m$ are i.i.d. Gaussian measurement matrices with i.i.d. standard normal entries, our objective is to recover the $k$-sparse signal $\bm{x}$ from measurements $y_i = \bm{x}^\top\bm{A}_i\bm{x}$ for $i = 1, \ldots, m$. It is worth noting that this model differs from phase retrieval, mainly due to the full-rank nature of the measurement matrices and its distinct applications. However, like phase retrieval, this problem suffers from the trivial ambiguity of a sign flip, and it is impossible to distinguish  ${\pm \bm{x}}$. Consequently, the error metric quantifying the distance between the obtained solution $\bm{\hat{x}}$ and the desired signal $\bm{x}$ is defined as $\min_{i=0,1}\{\|\bm{\hat{x}}-(-1)^i\bm{x}\|_2\}$

\item  \textbf{Generative Quadratic System.} 
A generative prior enforces that $\bm{x}$ lies in the range of an $L$-Lipschitz generative model $G(\cdot):\mathbb{B}^k(r)\to \mathbb{R}^n$, which is expressed as $\bm{x}\in G(\mathbb{B}^k(r))$. We also investigate the task of recovering such $\bm{x}$ (with a generative prior) from the quadratic system involving i.i.d. Gaussian matrices, as previously described. Following conventions that are more suitable for specific applications like image recovery \cite{liu2022generative, liu2022misspecified}, we assume that we can identify a point that is positively correlated with $\bm{x}$ to allow for the exact recovery of $\bm{x}$ (see Remark \ref{rem:trivial}).

\end{itemize}
 
\subsection{Prior Art}
We review the most relevant works that are categorized as follows.

\subsubsection{\textbf{The Rank-1 Model --- (Sparse) Phase Retrieval}}

For rank-one quadratic systems  $y_i=|\bm{a}_i^\top\bm{x}|^2,~i\in[m]$ with Gaussian $\bm{a}_i$, Cand\`{e}s et al. demonstrated that the global optimum can be found using the Wirtinger Flow (WF) algorithm \cite{candes2015phase}. Alternatively, this problem can be formulated as semidefinite programming via a convex relaxation technique known as PhaseLift \cite{Candes:PhaseLift:2013}. The original works of WF and PhaseLift   require a sample complexity of $m=O(n\log n)$, while subsequent developments have  reduced the sample complexity to $m=O(n)$ \cite{chen2017solving,candes2014solving}. However, these methods cannot handle the high-dimensional scenario where $m<n$, since they do not consider any signal structure a priori.

In the high-dimensional case,  the sparse phase retrieval problem where the target signal $\bm{x}$ is $k$-sparse ($k\ll n$) was extensively studied. For uniqueness or stable recovery, a measurement number on the order of $O\left(k\log(n/k)\right)$ is   sufficient \cite{ELDAR2014473,IWEN2017135,Huang:NPR:2020,XIA20211}. However, we note that many computationally feasible algorithms, such as PhaseLift \cite{Li:SparsePR:2013} and non-convex methods \cite{cai2016optimal,Wang:TAF:2018,Yuan:SWF:2019,jaganathan2017sparse}, typically require a sample complexity of $O(k^2\log n)$. Specifically, the aforementioned non-convex approaches often involve two steps: obtaining an good initialization and iteratively refining the solution. The initialization step frequently becomes a bottleneck due to its suboptimal sample complexity of $O(k^2\log n)$, whereas various update rules for refinement can succeed with (near-)optimal sample complexity.  It is possible to use an impractical initialization method with a sample complexity of $O(k\log n)$ \cite{Liu:CPR:2021}, but the possibility of a practical algorithm that requires only $O(k\log n)$ measurements remains an open question.

\subsubsection{\textbf{The Full-Rank Quadratic System}}

The primary focus of this work revolves around quadratic systems expressed as $y_i=\bm{x}^\top\bm{A}_i\bm{x}$ with full-rank matrices $\{\bm{A}_i\}_i$. In the low-dimensional case ($m>n$) without any signal structure a priori, several approaches that require an order-wise optimal     $O(n)$ (sub-)Gaussian measurements have been proposed in \cite{Huang:Quadratic:2020,Thaker:Quadratic:2020,Fan:RegNewton:2022}. In particular, \cite{Huang:Quadratic:2020} extended the WF algorithm \cite{candes2015phase} from phase retrieval to the full-rank case of interest. Such an extension is notably non-trivial and demands significant modifications in both algorithmic design and proof methodology, as highlighted in \cite{Huang:Quadratic:2020}.

The main objective of this paper is to tackle the high-dimensional scenario ($m< n$), where certain prior knowledge or structural assumptions about $\bm{x}$ are available. Arguably, one of the most classical signal structures is sparsity, yet even the comprehension of the sparse quadratic system (as detailed in Section \ref{sec:setup}) remains highly incomplete. To the best of our knowledge, there has not been a dedicated work specifically investigating this problem with an algorithm that can provably reconstruct the signal. Rather, the only related works are those studying compressed sensing of {\it simultaneously} (row-)sparse and low-rank matrices, which is a more general problem. To see this, through a technique known as ``lifting'', one can write $$y_i=\bm{x}^\top\bm{A}_i\bm{x}=\langle \bm{A}_i,\bm{xx}^\top\rangle$$ and thus formulate the sparse quadratic system as the linear compressed sensing of $\bm{X}=\bm{xx}^\top$, which is a rank-1 matrix with at most $k$ non-zero rows. As a result, some prior   works do offer some insights or even provable solutions to the sparse quadratic system problem. Specifically, Oymak et al. demonstrated that the most intuitive convex approach of combining two regularizers is highly sub-optimal \cite{Oymak:SLM:2015}. We will introduce a non-convex algorithm called thresholded Wirtinger flow (TWF) that can provably reconstruct the signal, whereas many existing non-convex methods fall short in  (global) convergence guarantees of this type (e.g., those in \cite{FORNASIER2021125702,Fan:PGM:2019,Fan:SRQuadratic:2018,Eisenmann2021RiemannianTM}). The most related existing algorithm (with convergence guarantee comparable to ours) is sparse power factorization (SPF) that can recover  matrix $\mX=\vu\vv^*$ with $k_1$-sparse $\bm{u}$ and $k_2$-sparse $\bm{v}$ from $y_i=\langle \bm{A}_i,\bm{X}\rangle$ using Gaussian $\bm{A}_i$ \cite{Lee:SPF:2018}. Thus, SPF addresses a more general problem and naturally provides a provable solution to the sparse quadratic system. However, compared to our TWF, SPF has certain drawbacks: (i) SPF's global recovery guarantee is applicable only to rapidly decaying $\bm{u}$ and $\bm{v}$ (satisfying $\|\bm{u}\|_\infty =\Omega(\|\bm{u}\|_2),\|\bm{v}\|_\infty =\Omega(\|\bm{v}\|_2)$), 
whereas our theoretical result for TWF applies to a general $k$-sparse $\bm{x}$; (ii) SPF updates $(\bm{u},\bm{v})$ through hard thresholding pursuit \cite{foucart2011hard} that involves iterations, whereas our TWF employs simple thresholded gradient descent with a closed-form update; (iii) SPF requires knowledge of the sparsity level, whereas TWF does not; (iv) SPF (or its subsequent developments \cite{Geppert:SPF:2019,Eisenmann2021RiemannianTM,FORNASIER2021125702}) is inherently sub-optimal for the sparse quadratic system problem since it does not exploit the information that $\bm{u}=\bm{v}=\bm{x}$. We will provide experimental results to demonstrate that TWF significantly outperforms SPF in solving sparse quadratic systems.

\subsubsection{\textbf{Generative Priors for Quadratic Systems}}
In recent years, deep generative models have shown tremendous success in a wide range of real-world applications, leading to a growing interest in replacing the commonly adopted sparsity assumption with generative modeling assumption for solving high-dimensional inverse problems. It has been widely demonstrated to be effective in reducing the number of measurements, see the seminal work~\cite{bora2017compressed} and many subsequent works (e.g., \cite{dhar2018modeling,liu2020sample,liu2022non}). Specifically, the use of generative priors in phase retrieval has been considered in~\cite{hand2018phase, hyder2019alternating, aubin2020exact, shamshad2020compressed, Liu:CPR:2021, killedar2022compressive, liu2022misspecified}, but there has not been any work that explores the application of generative modeling to a full-rank quadratic system. 

\subsection{Our Contribution}
\subsubsection{\textbf{Sparse Quadratic System}} 

To solve a sparse quadratic system, we introduce the Thresholded Wirtinger Flow (TWF) algorithm consisting of two  steps: spectral initialization (Algorithm \ref{alg1}) and thresholded gradient descent (Algorithm \ref{alg2}). TWF does not require prior knowledge of the sparsity level $k$. The spectral initialization step needs $m=O(k^2\log n)$ measurements to provide an initialization point close to $\bm{x}$ (up to a sign flip; see Lemma \ref{lem4}). 
In contrast, with a sufficiently good initialization, the thresholded gradient descent step requires only $m=O(k\log n)$ measurements for linear convergence (see Lemma \ref{lem5}). This intriguingly aligns with the requirements of various non-convex algorithms designed for sparse phase retrieval, such as those detailed in \cite{cai2016optimal,jaganathan2017sparse,Wang:TAF:2018}, where the sub-optimal $O(k^2\log n)$ is also incurred in the initialization phase (see Remark \ref{rem:ksqu}). Our theoretical analyses are considerably more intricate than the low-dimensional case explored in \cite{Huang:Quadratic:2020}. For instance, we need to introduce a ``new'' iteration sequence that is more amenable to analysis, followed by demonstrating that the TWF sequence coincides with this new sequence with high probability (see Remark \ref{rem:technical}). In our experimental evaluations, TWF delivers notable improvements over the existing SPF algorithm (see Figure \ref{fig:ptc}).

\subsubsection{\textbf{Generative Quadratic System}} 

To solve a generative quadratic system, we propose the Projected Gradient Descent (PGD) approach that analogously contains two steps: the Projected Power Method (Algorithm \ref{alg3}) for initialization and the Projected Gradient Descent (Algorithm \ref{alg4}) for refinement. Notably, our approach makes the first solution to this specific problem. While it is acknowledged that our algorithm may not be entirely practical, its two steps both achieve a measurement number that is nearly optimal and proportional to $k$, as opposed to the previously conventional $k^2$ for the sparse prior. This achievement extends the insights gained in \cite{Liu:CPR:2021} from the rank-1 case (i.e., phase retrieval) to the full-rank case (see Remark \ref{rem:genepp}). This extension, owing to the fundamentally different model, requires very different techniques.

We close this section with an outline of the paper. We first introduce the notations and preliminary concepts in Section \ref{sec:nota_pre}. Next, we delve into our primary algorithms and present our theoretical findings for both the sparse and generative cases in Sections \ref{sec:TWF} and \ref{sec:PGD}, respectively. Experimental results are detailed in Section \ref{sec:exp}, and we draw our conclusions in Section \ref{sec:conclude}. We defer the technical lemmas for the sparse case (Lemmas \ref{lem6}-\ref{lem12}) to the appendix. Due to space limit, the technical lemmas for the generative case (Lemmas 13-14), the proofs of the main results for the generative case (Theorems \ref{thm:init_gen}-\ref{thm:generefine}), and the proofs of Lemmas \ref{lem6}-14 are provided in the supplementary material.

\section{Notations and Preliminaries}\label{sec:nota_pre}

\subsection{Notations}
 
Throughout the paper, we use regular letters to denote scalars, and boldface symbols to denote vectors and matrices. We use the convention $[n]=\{1,...,n\}.$ We denote the $(j, \ell)$-th entry of $\bm{A}_i$ as $a^{(i)}_{j\ell}$. For simplicity, we define  $\bm{\widetilde{A}}_i= \frac{1}{2}\big(\bm{A}_i+\bm{A}_i^\top\big)$ and note the useful relation $y_i = \bm{x}^\top\bm{\widetilde{A}}_i\bm{x}$. For vectors $\bm{x},\bm{y}\in \mathbb{R}^n$, we work with the $\ell_2$-norm $\|\bm{x}\|_2$, the max-norm $\|\bm{x}\|_\infty=\max_i|x_i|$, the $\ell_1$ norm $\|\bm{x}\|_1$, support set $\supp(\bm{x})=\{i\in[n]:x_i\neq 0\}$, and inner product $\langle\bm{x},\bm{y}\rangle = \bm{x}^\top\bm{y}$. Further given $\mathcal{S}\subset [n]$, we let $\bm{x}_\mathcal{S}$ be the $n$-dimensional vector obtained from $\bm{x}$ by retaining entries indexed in $\mathcal{S}$ while setting all others to zero. We let $\Sigma^n_k$ be the set of all $k$-sparse signal in $\mathbb{R}^n$, $\mathbb{S}^{n-1}=\{\bm{x}\in \mathbb{R}^n:\|\bm{x}\|_2=1\}$ be the standard Euclidean sphere, and $\mathbb{B}^k(r)=\{\bm{x}\in \mathbb{R}^k:\|\bm{x}\|_2\leq r\}$. Given matrices $\bm{A},\bm{B}\in \mathbb{R}^{n\times n}$, we let $\|\bm{A}\|_F$ be the Frobenius norm, $\|\bm{A}\|_{op}$ be the operator norm (that equals the largest singular value), $\langle\bm{A},\bm{B}\rangle=\text{Tr}(\bm{A}^\top\bm{B})$ be the inner product. Further given $\mathcal{S},\mathcal{T}\subset [n]$, we set $[\bm{A}]_{\mathcal{S},\mathcal{T}}=[\hat{a}_{ij}]$, where $\hat{a}_{ij}=a_{ij}$ if $(i,j)\in \mathcal{S}\times \mathcal{T}$, and $\hat{a}_{ij}=0$ otherwise. Similarly to many important prior works in the area (e.g., \cite{candes2015phase,cai2016optimal}), our theoretical analysis does not refine multiplicative constants, and we use the generic notations $C,c,C_1,c_1,C_2,c_2,...$ to denote absolute constants whose values may vary per appearance. Given two terms $T_1$ and $T_2$, we write $T_1\lesssim T_2$ or $T_1=O(T_2)$ if $T_1\leq CT_2$ for some $C$, and conversely write $T_1\gtrsim T_2$ or $T_1=\Omega(T_2)$ if $T_1\geq cT_2$ for some $c$. We use $T_1\asymp T_2$ to state that $T_1=O(T_2)$ and $T_2=\Omega(T_2)$ simultaneously hold.  

\vspace{-1em}
\subsection{Preliminaries}
For a random variable $X$, we define the sub-Gaussian norm $$\|X\|_{\psi_2} = \inf\Big\{t>0:\mathbbm{E}\exp(\frac{X^2}{t^2})\leq 2\Big\}$$ and sub-exponential norm $$\|X\|_{\psi_1} = \inf\Big\{t>0:\mathbbm{E}\exp(\frac{|X|}{t})\leq 2\Big\}.$$ The $X$ with  finite $\|X\|_{\psi_2}$ is said to be sub-Gaussian. Note that the sub-Gaussian $X$ exhibits an exponentially decaying probability tail, i.e., for any $t>0$, $$
        \mathbbm{P}(|X|\geq t)\leq 2\exp\Big(-\frac{ct^2}{\|X\|_{\psi_2}^2}\Big).$$ Similarly, $X$ with finite $\|X\|_{\psi_1}$ is sub-exponential and has the   tail bound  $$
     \mathbbm{P}(|X|\geq t)\leq 2\exp\Big(-\frac{ct}{\|X\|_{\psi_1}}\Big) $$ for any $t>0$. 
To relate sub-Gaussian norm and sub-exponential norm, for sub-Gaussian variables $X$ and $Y$ one has (see, e.g., \cite[Lem. 2.7.7]{vershynin2018high})\begin{equation}\label{2.3}
    \|XY\|_{\psi_1}\leq \|X\|_{\psi_2}\|Y\|_{\psi_2}.
\end{equation}
For an $n$-dimensional random vector $\bm{X}$, we let $\|\bm{X}\|_{\psi_2}=\sup_{\bm{v}\in\mathbb{S}^{n-1}}\|\bm{v}^\top\bm{X}\|_{\psi_2}$. The following Bernstein's inequality for the concentration of sub-exponential random variables will be recurring in our analyses.
\begin{lem}\label{bernstein}
    {\rm(Bernstein's inequality, \cite[Thm. 2.8.1]{vershynin2018high})} Let $X_1,...,X_N$ be independent, zero-mean, sub-exponential random variables, then for every $t\geq 0$ and some absolute constant $c$, we let $A=\sum_{i=1}^N\|X_i\|^2_{\psi_2}$, $B=\max_{1\leq i\leq N}\|X_i\|_{\psi_1}$ and have $$            \mathbbm{P}\big(\big|\sum_{i=1}^N X_i\big|\geq t\big)\leq 2\exp\big(-C\min\big\{\frac{t^2}{A},\frac{t}{B}\big\}\big).$$
\end{lem}
As in many prior works (e.g., \cite{chen2017solving}), the Davis-Kahan Theorem will be useful in analysing spectral initialization.
\begin{lem}
    \label{Davis-Kahan}{\rm (Davis-Kahan $\sin\bm{\Theta}$ Theorem, e.g., \cite[Thm. 20]{chi2019nonconvex})} Let $\bm{Y},\bm{Y}_{\star}\in \mathbb{R}^{n\times n}$ be two symmetric matrices and we order the eigenvalues of $\bm{Y}_{\star}$ as $\lambda_1(\bm{Y}_{\star})\geq ...\geq \lambda_{r}(\bm{Y}_{\star})>\lambda_{r+1}(\bm{Y}_{\star})\geq ...\geq \lambda_n(\bm{Y}_{\star})$. Denote by $\bm{U}_{\star}\in \mathbb{R}^{n\times r}$ (resp. $\bm{U}\in \mathbb{R}^{n\times r}$) an orthonormal matrix whose columns are the first $r$ eigenvectors of $\bm{Y}_{\star}$ (resp. $\bm{Y}$), then if $\|\bm{Y}-\bm{Y}_{\star}\|_{op}< \lambda_r(\bm{Y}_{\star})-\lambda_{r+1}(\bm{Y}_{\star})$, we have 
    \begin{equation}\nonumber
        \|\bm{UU}^\top-\bm{U}_\star\bm{U}_\star^\top\|_{op}\leq \frac{\|\bm{Y}-\bm{Y}_{\star}\|_{op}}{\lambda_r(\bm{Y}_{\star})-\lambda_{r+1}(\bm{Y}_{\star})-\|\bm{Y}-\bm{Y}_{\star}\|_{op}}.
    \end{equation}
\end{lem}

\section{Solving Sparse Quadratic System via Thresholded Wirtinger Flow}\label{sec:TWF}
We propose our first algorithm ``thresholded Wirtinger flow (TWF)'' for solving sparse quadratic systems.\footnote{We focus on the real case here. We expect that the extension to the complex case is straightforward and still refer to the algorithm as TWF.}  Given a threshold $\tau>0$, we let $\mathcal{T}_\tau(\cdot)$ be the thresholding operator satisfying
\begin{equation}\nonumber
    \mathcal{T}_\tau(a)=0~\text{if}~|a|\leq \tau; \textnormal{ and }~|\mathcal{T}_\tau(a)-a|\leq \tau~\text{if}~|a|>\tau. 
\end{equation}
In particular, $\mathcal{T}_\tau(\cdot)$ can be either the standard hard thresholding or the soft thresholding operator. 

\subsection{Algorithm: Initialization and Refinement}
Our TWF algorithm consists of two steps, namely spectral initialization and thresholded gradient descent. We first present spectral initialization in Algorithm \ref{alg1}, which aims to provide an initial guess that is sufficiently close to $\bm{x}$. The intuition of using the leading eigenvector of (a sub-matrix of) $\bm{\widetilde{S}}_{in}$ is based on $\mathbbm{E}[\bm{\widetilde{S}}_{in}]=\bm{xx}^\top$ (Lemma \ref{lem2}), which resembles Equation (13) in \cite{Huang:Quadratic:2020} that studied the low-dimensional case ($m\geq n$). Nonetheless, addressing the high-dimensional sparse case requires additional nontrivial effort, specifically the support estimate in (\ref{defi:hatS_0}) and the restriction of the data matrix $\bm{\widetilde{S}}_{in}$ to $[\bm{\widetilde{S}}_{in}]_{\widehat{S}_0,\widehat{S}_0}$; the technical benefit of the latter is to allow for faster operator norm concentration (see Lemma \ref{lem6}). 
 \begin{algorithm} 
\caption{Spectral Initialization for TWF}\label{alg1}
\begin{algorithmic}[1]
 \Statex \textbf{Input:} $(\bm{A}_i,y_i)_{i=1}^n$; tuning parameter $\alpha$

 \item  Compute $\phi:=\left(\frac{1}{m}\sum_{i=1}^m y_i^2\right)^{1/4}$ to approximate $\|\bm{x}\|_2$ 

\item  
Compute $
    I_l = \frac{1}{m}\sum_{i=1}^m y_ia^{(i)}_{ll}$ for $l\in [n]$ and use the following to approximate $\supp(\bm{x})$: \begin{equation}\label{defi:hatS_0}
    \widehat{S}_0=\left\{l\in [n]:I_l > \alpha \phi^2\sqrt{\frac{\log n}{m}}\right\}.
\end{equation}

\item Construct an $n\times n$ matrix $\bm{S}_{in}=\frac{1}{m}\sum_{i=1}^m y_i\bm{A}_i$ and then symmetrize it to
$\bm{\widetilde{S}}_{in} = \frac{1}{2}\big(\bm{S}_{in}+\bm{S}_{in}^\top\big)=\frac{1}{m}\sum_{i=1}^m y_i\bm{\widetilde{A}}_i$. Then compute $\bm{\hat{v}}_1$ as the normalized leading eigenvector of $[\bm{\widetilde{S}}_{in}]_{\widehat{S}_0,\widehat{S}_0}$.
\Statex\textbf{Output:} $\bm{\hat{x}}_0=\phi \bm{\hat{v}}_1$.
\end{algorithmic}
\end{algorithm}

Next, we present TWF updates that further refine $\bm{\hat{x}}_0$. Analogously to the algorithm in \cite{Huang:Quadratic:2020}, it is based on the goal of minimizing the   $\ell_2$ loss 
$
    f(\bm{z})= \frac{1}{4m}\sum_{i=1}^m \big(\bm{z}^\top\bm{A}_i\bm{z}-y_i\big)^2$ 
via gradient descent,   and the gradient is given by\begin{equation}\label{defi:gradient}
    \nabla f(\bm{z})= \frac{1}{m}\sum_{i=1}^m\big(\bm{z}^\top\bm{A}_i\bm{z}-y_i\big)\bm{\widetilde{A}}_i\bm{z}.
\end{equation}
Furthermore, as developed in \cite{cai2016optimal} for sparse phase retrieval, we employ a thresholding step after each gradient descent step to incoporate the sparse prior into the recovery. The details are given in Algorithm \ref{alg2}.
\begin{algorithm} 
\caption{Thresholded Wirtinger Flow}\label{alg2}
\begin{algorithmic}[1]
 \Statex \textbf{Input:}   $(\bm{A}_i,y_i)_{i=1}^n$; thresholding operator $\mathcal{T}$;  step size $\mu$;   tuning parameter $\beta$; number of iterations $T$.

 \item Run Algorithm \ref{alg1} to get $\bm{\hat{x}}_0$. 

 \item \textbf{for} $t=0,1,...,T-1$

 Compute thresholded level at {$\bm{\hat{x}}_{t}$} 
\begin{equation}\label{235}
    \tau(\bm{\hat{x}}_t)=\Big[{\frac{\beta}{m^2}\sum_{i=1}^m\big(\bm{\hat{x}}_{t}^\top\bm{A}_i\bm{\hat{x}}_t-y_i\big)^2}\Big]^{1/2}\cdot\|\bm{\hat{x}}_t\|_2.
\end{equation}

Compute $\nabla f(\bm{\hat{x}}_t)$ as in (\ref{defi:gradient}), then update \begin{equation}
    \label{defi:twfupdate}\bm{\hat{x}}_{t+1}=\varphi(\bm{\hat{x}}_t):=\mathcal{T}_{\frac{\mu}{\phi^2}\tau(\bm{\hat{x}}_t)}\Big(\bm{\hat{x}}_t-\frac{\mu}{\phi^2}\nabla f(\bm{\hat{x}}_t)\Big).
\end{equation}

\noindent
\textbf{end}

\Statex\textbf{Output:} $\bm{\hat{x}}= \bm{\hat{x}}_{T}$.
\end{algorithmic}
\end{algorithm}

\subsection{Theory: Linear Convergence Guarantee} 
In the proof, without loss of generality, we assume that the 
$k$-sparse signal $\bm{x}$ satisfies $\|\bm{x}\|_2=1$ and $\supp(\bm{x})\subset S=[k]$.

\subsubsection{Analysing a variant using support knowledge} Our analyses of TWF resemble the strategy developed in \cite{cai2016optimal}, that is, we first analyze an ideal impractical algorithm requiring the support knowledge $\mathrm{supp}(\bm{x})\subset S=[k]$, and then prove that our TWF algorithm provides the same iteration sequence with high probability. In particular, we will first concentrate on an algorithm (referred to as TWF with support knowledge and abbreviated as TWF-S)  that produces an iteration sequence $\{\bm{\tilde{x}}_i:i=0,1,\ldots,T\}$ as follows:

\begin{itemize}
    \item \textbf{Initialization in TWF-S.} Let $\alpha$, $\phi$, $I_l$, $\bm{S}_{in}$ and $\bm{\widetilde{S}}_{in}$ be the same as in Algorithm \ref{alg1}. Compared to $\widehat{S}_0$ in (\ref{defi:hatS_0}), we further restrict it to $S=[k]$ and use\footnote{It will be clear in the proof of Lemma \ref{lem4} that $\widetilde{S}_0$ is non-empty with high probability, specifically see (\ref{eq:non_empty_set}) and (\ref{eq:non_empty_2}) therein.}  
    \begin{equation}\label{defi:Sselect}
        \widetilde{S}_0=\widehat{S}_0\cap[k]=\Big\{l\in [k]:I_l>\alpha\phi^2\sqrt{\frac{\log n}{m}}\Big\}.
    \end{equation}
    Then, we compute $\bm{\tilde{v}}_1$ as the normalized leading eigenvector of $[\bm{\widetilde{S}}_{in}]_{\widetilde{S}_0,\widetilde{S}_0}$, and use $\bm{\tilde{x}}_0:=\phi \bm{\tilde{v}}_1$ as the initialization.
    \item \textbf{Updates in TWF-S.} Let $\mathcal{T}$, $\mu$, $\beta$ be the same as Algorithm \ref{alg2}. 
    For $t=0,1,...,T-1$, we compute the threshold level at the current iteration point $\bm{\tilde{x}}_t$ in the same way as (\ref{235})
    \begin{equation}\label{defi:gradthre}
        \tau(\bm{\tilde{x}}_t)=\Big[{\frac{\beta}{m^2}\sum_{i=1}^m \big(\bm{\tilde{x}}_t^\top\bm{A}_i\bm{\tilde{x}}_t-y_i\big)^2}\Big]^{1/2}\cdot\|\bm{\tilde{x}}_t\|_2.
    \end{equation}
    We then update $\bm{\tilde{x}}_t$ to $\bm{\tilde{x}}_{t+1}$ using the thresholded gradient descent that differs from (\ref{defi:twfupdate}) by further restricting the gradient to $S=[k]$:
    \begin{equation}
\label{defi:Supdate}\bm{\tilde{x}}_{t+1}=\eta(\bm{\tilde{x}}_t):=\mathcal{T}_{\frac{\mu}{\phi^2}\tau(\bm{\tilde{x}}_t)}\Big(\bm{\tilde{x}}_t-\frac{\mu}{\phi^2}\nabla f(\bm{\tilde{x}}_t)_S\Big).
    \end{equation}
\end{itemize}
An important observation is the independence between the iteration sequence $\{\bm{\tilde{x}}_t\}_{t=0}^T$ and the majority of $\bm{A}_i$'s entries. 
\begin{pro}\label{pro1}
    The sequence   $\{\bm{\tilde{x}}_t:t=0,1,...,T\}$ is independent of $E':=\{a^{(l)}_{ij}:l\in [m],(i,j)\notin [k]\times [k]\}$.
\end{pro}
\begin{proof}
     \textbf{(Initialization)}
    Because $\mathrm{supp}(\bm{x})\subset S=[k]$, we have $y_i=\bm{x}^\top\bm{A}_i\bm{x}=\bm{x}_S^\top[\bm{A}_i]_{S,S}\bm{x}_S$, which indicates that $\{y_i:i\in [m]\}$ is independent of $E'$. Thus, the quantities $\{\phi,I_1,...,I_k\}$ calculated in Algorithm \ref{alg1}, by their definitions, are independent of $E'$, which implies that $\widetilde{S}_0$ defined in (\ref{defi:Sselect}) is independent of $E'$. Because $\widetilde{S}_0\subset S$,   we can write \begin{equation}
        \begin{aligned}
            \nonumber[\bm{\widetilde{S}}_{in}]_{\widetilde{S}_0,\widetilde{S}_0}&=\frac{1}{m}\sum_{i=1}^m\big[\big[y_i\bm{\widetilde{A}}_i\big]_{S,S}\big]_{\widetilde{S}_0,\widetilde{S}_0}\\&=\frac{1}{m}\sum_{i=1}^m\big[y_i\big([\bm{A}_i]_{S,S}^\top+[\bm{A}_i]_{S,S}\big)\big]_{\widetilde{S}_0,\widetilde{S}_0},
        \end{aligned}
    \end{equation}
    indicating that $[\bm{\widetilde{S}}_{in}]_{\widetilde{S}_0,\widetilde{S}_0}$ is independent of $E'$. Moreover, since $\bm{\tilde{x}}_0=\phi \bm{\tilde{v}}_1$ with $\bm{\tilde{v}}_1$ being the leading eigenvector of $[\bm{\widetilde{S}}_{in}]_{\widetilde{S}_0,\widetilde{S}_0}$, we know $\bm{\tilde{x}}_0$ is independent of $E'$.

    \noindent{\textbf{(Updates)}}
    It remains to show $\{\bm{\tilde{x}}_1,...,\bm{\tilde{x}}_T\}$ is independent of $E'$, for which we use induction. Let $t\in\{0,1,...,T-1\}$, we assume $\bm{\tilde{x}}_t$ satisfying $\supp(\bm{\tilde{x}}_t)\subset S =[k]$ is independent of $E'$ (note that this is already satisfied by $\bm{\tilde{x}}_0$), and we aim to show $\bm{\tilde{x}}_{t+1}$ satisfying $\mathrm{supp}(\bm{\tilde{x}}_{t+1})\subset S$ is also independent of $E'$. Firstly, note that $\mathrm{supp}(\bm{\tilde{x}}_t)\subset S$,  so $\tau(\bm{\tilde{x}}_t)$ in (\ref{defi:gradthre}) is independent of $E'$ since it can be written as 
    $$
        \tau(\bm{\tilde{x}}_t)=\sqrt{\frac{\beta}{m^2}\sum_{i=1}^m(\bm{\tilde{x}}_t^\top[\bm{A}_i]_{S,S}\bm{\tilde{x}}_t-y_i)^2}\cdot\|\bm{\tilde{x}}_t\|_2.$$
    Secondly, by (\ref{defi:gradient}) and $\supp(\bm{\tilde{x}}_t)\subset S=[k]$, we have \begin{equation}
        \begin{aligned}\nonumber 
           \nabla f(\bm{\tilde{x}}_t)_S&=\frac{1}{m}\sum_{i=1}^m\big(\bm{\tilde{x}}_t^\top\bm{A}_i\bm{\tilde{x}}_t-y_i\big)[\bm{\widetilde{A}}_i\bm{\tilde{x}}_t]_S\\&=\frac{1}{m}\sum_{i=1}^m\big(\bm{\tilde{x}}_t^\top[\bm{A}_i]_{S,S}\bm{\tilde{x}}_t-y_i\big)[\bm{\widetilde{A}}_i]_{S,S}\bm{\tilde{x}}_t.
        \end{aligned}
    \end{equation}
    Thus, $\nabla f(\bm{\tilde{x}}_t)_S$ is independent of $E'$. Combined with (\ref{defi:Supdate}), we know that $\bm{\tilde{x}}_{t+1}$ is independent of $E'$, and evidently $\mathrm{supp}(\bm{\tilde{x}}_{t+1})\subset S$, which completes the proof. 
\end{proof}

\begin{remark}\label{rem:technical}
    {\rm (Proof strategy and technical differences with \cite{cai2016optimal,Huang:Quadratic:2020})} The thresholded Wirtinger flow algorithm for phase retrieval was developed in \cite{cai2016optimal}, but since we are considering the full-rank case, the random terms arising in our analyses are totally different from those in \cite{cai2016optimal} and should be bounded by different techniques. Compared with the Wirtinger flow for a low-dimensional full-rank quadratic system in \cite{Huang:Quadratic:2020}, our analyses are much more involved. Most notably, unlike \cite{Huang:Quadratic:2020} that directly analyzes the Wirtinger flow sequence, we need to first analyze the two steps of TWF-S in Lemmas \ref{lem4}-\ref{lem5} (note that the independence described in Proposition \ref{pro1} makes this essentially easier than directly analyzing the TWF sequence), and then show both steps of TWF-S and TWF coincide with high probability (Lemmas \ref{lem8}-\ref{lem7}). In addition, compared to \cite{Huang:Quadratic:2020}, our TWF needs extra steps to incorporate the sparse prior, which in general leads to more random terms of various sorts (e.g., comparing our Lemma \ref{lem4} and \cite[Lem. 1]{Huang:Quadratic:2020}). 
\end{remark}

\subsubsection*{\textbf{Spectral Initialization in TWF-S}}

Fix any $\delta_0>0$, we will show that $\min_{i=0,1}\|\bm{\tilde{x}}_0-(-1)^i\bm{x}\|_2\leq \delta_0$ holds with high probability. We get started with several lemmas. 
\begin{lem}\label{lem2}
     $\mathbbm{E}[\bm{\widetilde{S}}_{in}]=\mathbbm{E}[\bm{S}_{in}]=\bm{xx}^\top$.
\end{lem}
\begin{proof}
    Evidently, $\mathbbm{E}[\bm{\widetilde{S}}_{in}]=\mathbbm{E}[\bm{S}_{in}]=\mathbbm{E}[y_i\bm{A}_i]$. Furthermore, the $(p,q)$-th entry of $S$ can be calculated as 
    $$
        \mathbbm{E}\big[\sum_{j,l\in [n]}x_jx_la^{(i)}_{jl}a^{(i)}_{pq}\big]=x_px_q,$$
    thus yielding the claim.
\end{proof}
Next, we establish two element-wise concentration results by Bernstein's inequality (Lemma \ref{bernstein}). The first one states that $\phi= (\frac{1}{m}\sum_{i=1}^my_i^2)^{1/4}$ provides a faithful estimate of the signal norm $\|\bm{x}\|_2$. 
\begin{lem}
    \label{lem1}
    For any $\delta\in (0,1)$, we have 
    $
        \mathbbm{P}( 1-\delta\leq \phi \leq 1+\delta^{1/4})\leq 2\exp(-cm\delta^2).$
    In particular, taking $\delta = 10^{-4}$ yields that   $0.9\leq \phi\leq 1.1$ holds with probability at least $1-2\exp(-\Omega(m))$.
\end{lem}
\begin{proof}
    By $\phi = \big(\frac{1}{m}\sum_{i=1}^my_i^2\big)^{1/4}$, we know $|\phi^4-1|=\big|\frac{1}{m}\sum_{i=1}^m y_i^2-1\big|$. Because $$y_i=\bm{x}^\top\bm{A}_i\bm{x}=\langle\bm{A}_i,\bm{x}\bm{x}^\top\rangle\sim \mathcal{N}(0,\|\bm{x}\bm{x}^\top\|_F^2)$$ and $$\|\bm{x}\bm{x}^\top\|_F^2=\textrm{Tr}(\bm{x}\bm{x}^\top \bm{x}\bm{x}^\top)=\|\bm{x}\|_2^2 \textrm{Tr}(\bm{x}\bm{x}^\top)=\|\bm{x}\|_2^4=1,$$ we know $y_i\sim\mathcal{N}(0,1)$ and hence $\|y_i^2\|_{\psi_1}\leq \|y_i\|_{\psi_2}^2=O(1)$. Thus, by using Lemma \ref{bernstein}, we obtain for any $t\geq 0$ that $$\mathbbm{P}(|\phi^4-1|\geq t)\leq 2\exp(-cm\min\{t,t^2\}).$$
    We set $t=\delta\in(0,1)$ and note that $|\phi^4-1|\leq \delta$ implies $1-\delta\leq \phi\leq 1+\delta^{1/4}$, thus proving the first claim.   Taking $\delta=10^{-4}$ yields the statement after "in particular", thus completing the proof.
\end{proof}
Our second element-wise concentration result guarantees that $I_l$'s are (uniformly) close to $x_l^2$'s, which justifies that we estimate $\supp(\bm{x})$ by taking the coordinates corresponding to large $I_l$s in (\ref{defi:hatS_0}) and (\ref{defi:Sselect}). 
\begin{lem}\label{lem3}
    When $m=\Omega(\log n)$, there exists some sufficiently large $C>0$, such that with probability at least $1-2n^{-9}$ we have 
    \begin{equation}\label{3.11}
        |I_l-x_l^2|\leq C\sqrt{\frac{\log n}{m}},~l=1,2,...,n.
    \end{equation}
\end{lem}
\begin{proof}
Observe that $I_l$ is simply the $(l,l)$-th entry of $\bm{S}_{in}$, from Lemma \ref{lem2} we know $\mathbbm{E}[I_l]=x_l^2$. Note that $I_l=\frac{1}{m}\sum_{i=1}^m y_ia_{ll}^{(i)}$, and $\|y_ia_{ll}^{(i)}\|_{\psi_1}\leq \|y_i\|_{\psi_2}\|a_{ll}^{(i)}\|_{\psi_2}=O(1)$, so Lemma \ref{bernstein} yields for any $t
\geq 0$ that $$
    \mathbbm{P}\big(|I_l-\mathbbm{E}[I_l]|\geq t\big)\leq 2\exp\big(-{c_1}m\min\{t,t^2\}\big).$$
Furthermore, we take a union bound over $l\in[n]$ to obtain that for any $t\geq 0$, \begin{align}
    \mathbbm{P}\Big(\sup_{l\in [n]}|I_l-x_l^2|\geq t\Big)\leq 2\exp\big(\log n-{c_1}m\min\{t,t^2\}\big).\label{eq:bernstein11}
\end{align}
Then we set \begin{align}\label{eq:choose_t}
    t= \sqrt{\frac{10 \log n}{c_1m}},
\end{align} and note that  under $m=\Omega(\log n)$ we can assume $t^2<t$. Substituting (\ref{eq:choose_t}) into (\ref{eq:bernstein11}) yields that 
\begin{align}
    \sup_{l\in[n]}|I_l-\mathbbm{E}[I_l]|<\sqrt{\frac{10\log n}{c_1m}}
\end{align}
holds with probability at least $1-2n^{-9}$, as claimed. 
\end{proof}

We are now ready to show that the spectral initialization provides $\min_{i=0,1}\|\bm{\tilde{x}}_0-(-1)^i\bm{x}\|_2\leq \delta_0$ with high probability.
\begin{lem}\label{lem4}
    Fix any $\delta_0\in (0,\frac{1}{2})$, there exists some $c_1$   only depending on $\delta_0$, when $m\geq c_1 k^2\log n$,    with probability at least $1-2\exp(-c_2m)-2\exp(-k)-2n^{-9}$ we have $$\min_{i=0,1}\|\bm{\tilde{x}}_0-(-1)^i\bm{x}\|_2\leq \delta_0.$$    
\end{lem}
\begin{proof}

Upon rescaling $\delta_0$, we only need to show $$\min_{i=0,1}\|\bm{\tilde{x}}_0-(-1)^i\bm{x}\|_2=O(\delta_0).$$ With probability at least $1-2n^{-9}-2\exp(-k)$, we proceed with our analysis under (\ref{3.11})  and (\ref{4.1}). In this proof,  we reserve $c$ for the absolute constant such that (\ref{3.11}) holds.
    By Lemma \ref{lem1}, we know that with high probability, $\phi\in [0.9,1.1]$. Thus, with a properly chosen $\alpha$ we assume $\alpha \phi^2\in (c,2c)$. We consider the index set \begin{align}
        \widetilde{S}_0'=\Big\{l\in[k]:x_l^2\geq 3c\sqrt{\frac{\log n}{m}}\Big\},\label{eq:non_empty_set}
    \end{align}
        {which is non-empty when $m\ge c_1k^2\log n$ with sufficiently large $c_1$ (otherwise, we have $\|\bm{x}\|_2^2\le 3ck\sqrt{\frac{\log n}{m}}<1$).}
    Note that for $l\in\widetilde{S}_0'$, combined with (\ref{3.11}), it holds that $$I_l\geq x_l^2- |I_l-x_l^2|\geq 2c\sqrt{\frac{\log n}{m}}>\alpha\phi^2 \sqrt{\frac{\log n}{m}},$$ thereby producing $l\in \widetilde{S}_0$. We then have \begin{align}\label{eq:non_empty_2}
        {\widetilde{S}_0'\subset \widetilde{S}_0\subset S.}
    \end{align} To bound $\min_{i=0,1}\|\bm{\tilde{x}}_0-(-1)^i\bm{x}\|_2$, we begin with \begin{equation}
        \begin{aligned}\label{3.15}
            &\min_{i=0,1}\|\bm{\tilde{x}}_0-(-1)^i\bm{x}\|_2\leq \big\|\bm{\tilde{x}}_0-\bm{\tilde{v}}_1\big\|_2\\&\quad +\min_{i=0,1}\big\|\bm{\tilde{v}}_1-(-1)^i\bm{x}_{\widetilde{S}_0}\big\|_2+\|\bm{x}_{\widetilde{S}_0}-\bm{x}\|_2.
        \end{aligned}
    \end{equation}
    Next, Let us bound the above three terms separately.

    \noindent{\textbf{(Bounding $\|\bm{\tilde{x}}_0-\bm{\tilde{v}}_1\|_2$)}}
    Firstly,  by the first statement in Lemma \ref{lem1}, with probability at least $1-2\exp(-\Omega(\delta_0^8m))$ we have $|\phi-1|\leq \delta_0$, which implies $$\|\bm{\tilde{x}}_0-\bm{\tilde{v}}_1\|_2=|\phi-1|\|\bm{\tilde{v}}_1\|_2\leq \delta_0.$$

    \noindent{\textbf{(Bounding $\|\bm{x}_{\widetilde{S}_0}-\bm{x}\|_2$)}}
        We now switch to the third term in (\ref{3.15}). By using $\widetilde{S}_0'\subset \widetilde{S}_0\subset S$, we have \begin{equation}
            \begin{aligned}\label{eq:inisub}
                &\|\bm{x}_{\widetilde{S}_0}-\bm{x}\|^2_2\leq \|\bm{x}_{\widetilde{S}_0'}-\bm{x}\|^2_2\\&= \sum_{i\in S}x_i^2\mathbbm{1}\Big(x_i^2<3c \sqrt{\frac{\log n}{m}}\Big)\leq 3ck\sqrt{\frac{\log n}{m}}.
            \end{aligned}
        \end{equation}
        Hence, when $m\gtrsim k^2\log n$, with a sufficiently large hidden constant, it holds that $\|\bm{x}_{\widetilde{S}_0}-\bm{x}\|_2\leq \delta_0$. This also implies $$\|\bm{x}_{\widetilde{S}_0}\|_2\geq \|\bm{x}\|_2-\|\bm{x}-\bm{x}_{\widetilde{S}_0}\|_2 \geq 1-\delta_0.$$

    \noindent{\textbf{(Bounding $\min_{i=0,1}\|\bm{\tilde{v}}_1-(-1)^i \bm{x}_{\widetilde{S}_0}\|_2$)}}
    We begin by bounding this term as $$\min_{i=0,1}\Big\|\bm{\tilde{v}}_1-(-1)^i\frac{\bm{x}_{\widetilde{S}_0}}{{\|\bm{x}_{\widetilde{S}_0}}\|_2}\Big\|_2+\Big\|\frac{\bm{x}_{\widetilde{S}_0}}{{\|\bm{x}_{\widetilde{S}_0}}\|_2}-\bm{x}_{\widetilde{S}_0}\Big\|_2,$$ and the second term here equals  $|1-\|\bm{x}_{\widetilde{S}_0}\|_2| $ and is hence bounded by $\delta_0$, as implied by the last part. Since $\|\bm{\tilde{v}}_1\|_2=1$, we can obtain that \begin{equation}
        \begin{aligned}\nonumber
            \min_{i=0,1}\Big\|\bm{\tilde{v}}_1-(-1)^i \frac{\bm{x}_{\widetilde{S}_0}}{\|\bm{x}_{\widetilde{S}_0}\|_2}\Big\|_2&= \sqrt{2-2\Big|\frac{\bm{\tilde{v}}_1^\top\bm{x}_{\widetilde{S}_0}}{\|\bm{x}_{\widetilde{S}_0}\|_2}\Big|}\\
            \leq \sqrt{2-2\Big|\frac{\bm{\tilde{v}}_1^\top\bm{x}_{\widetilde{S}_0}}{\|\bm{x}_{\widetilde{S}_0}\|_2}\Big|^2}&=\Big\|\bm{\tilde{v}}_1\bm{\tilde{v}}_1^\top - \frac{\bm{x}_{\widetilde{S}_0}\bm{x}_{\widetilde{S}_0}^\top}{\|\bm{x}_{\widetilde{S}_0}\|_2^2}\Big\|_F\\
           &\leq \sqrt{2}\Big\|\bm{\tilde{v}}_1\bm{\tilde{v}}_1^\top - \frac{\bm{x}_{\widetilde{S}_0}\bm{x}_{\widetilde{S}_0}^\top}{\|\bm{x}_{\widetilde{S}_0}\|_2^2}\Big\|_{op}.
        \end{aligned}
    \end{equation} Recall that $\bm{\tilde{v}}_1$ is the normalized leading eigenvector of $[\bm{\widetilde{S}}_{in}]_{\widetilde{S}_0,\widetilde{S}_0}$, while evidently, $\bm{x}_{\widetilde{S}_0}/\|\bm{x}_{\widetilde{S}_0}\|_2$ is the normalized leading eigenvector of $\mathbbm{E}[\bm{\widetilde{S}}_{in}]_{\widetilde{S}_0,\widetilde{S}_0}=\bm{x}_{\widetilde{S}_0}\bm{x}_{\widetilde{S}_0}^\top$, whose unique non-zero eigenvalue is $\|\bm{x}_{\widetilde{S}_0}\|_2^2$. Moreover, because $\widetilde{S}_0\subset S$, the  concentration error of $[\bm{\widetilde{S}}_{in}]_{\widetilde{S}_0,\widetilde{S}_0}$ is controlled by Lemma \ref{lem6} that provides $$\|[\bm{\widetilde{S}}_{in}]_{\widetilde{S}_0,\widetilde{S}_0}-\bm{x}_{\widetilde{S}_0}\bm{x}^\top_{\widetilde{S}_0}\|_{op}\leq 
        \big\|[\bm{\widetilde{S}}_{in}]_{S,S}-\bm{xx}^\top\big\|_{op}\leq c_1 \sqrt{\frac{k}{m}}$$ with the promised probability.
        These allow us to invoke the Davis-Kahan Theorem (Lemma \ref{Davis-Kahan}) and obtain \begin{equation}\nonumber\begin{aligned}
          \Big\|\bm{\tilde{v}}_1\bm{\tilde{v}}_1^\top - \frac{\bm{x}_{\widetilde{S}_0}\bm{x}_{\widetilde{S}_0}^\top}{\|\bm{x}_{\widetilde{S}_0}\|_2^2}\Big\|_{op}&\leq \frac{c_1\sqrt{k/m}}{\|\bm{x}_{\widetilde{S}_0}\|_2^2-c_1\sqrt{k/m}}\\&\stackrel{(i)}{\leq} \frac{c_1\sqrt{k/m}}{(1-\delta_0)-c_1\sqrt{k/m}} \stackrel{(ii)}{\leq} \delta_0
          \end{aligned}
        \end{equation}
where we recall $\|\bm{x}_{\widetilde{S}_0}\|_2\geq 1-\delta_0$ in $(i)$ and use $m\gtrsim k$ for some hidden constant depending on $\delta_0$ in $(ii)$. 
         By substituting these bounds into  (\ref{3.15}), we have shown $\min_{i=0,1}\|\bm{\tilde{x}}_0-(-1)^i\bm{x}\|_2=O(\delta_0)$, as desired.    
\end{proof}

\begin{remark}
    \label{rem:ksqu}
    {\rm (Initialization incurs the sub-optimal $O(k^2\log n)$)}  Our spectral initialization requires a sample size of $m = O(k^2\log n)$ that is significantly larger than the information-theoretic optimal $O(k\log n)$ (e.g., \cite{oymak2015near}). In contrast, as shown in Lemma \ref{lem5} below, the thresholded gradient descent succeeds with the near-optimal $m=O(k\log n)$. We note that: (i) for sparse phase retrieval, the initialization steps of most non-convex algorithms (e.g., \cite{cai2016optimal,jaganathan2017sparse,Wang:TAF:2018}) encounter the same gap, and it remains a significant open question whether the gap can be closed by a tractable algorithm; (ii) for our full-rank case, the SPF algorithm in \cite{Lee:SPF:2018}  succeeds with a near-optimal $m=O(k\log\frac{n}{k})$ but only for a very restricted set of $k$-sparse $\bm{x}$ satisfying $\|\bm{x}\|_\infty = \Omega(\|\bm{x}\|_2)$ \cite[Thm. 1]{Lee:SPF:2018}. Therefore, the counterpart of a significant open question in the rank-1 sparse phase retrieval problem was found in our full-rank sparse quadratic system. The possibility of refining 
 $m=O(k^2\log n)$ for a {\it general} $k$-sparse $\bm{x}$ is left for future research. 
\end{remark}
\subsubsection*{\textbf{Thresholded Gradient Descent in TWF-S}}

\begin{lem}\label{lem5}
Consider a fixed $\bm{x}\in \mathbb{S}^{n-1}$ satisfying $\supp(\bm{x})\subset S = [k]$ as the desired signal. If $m\gtrsim k\log n$, then with probability at least $1-C_1\exp(-\Omega(k))$ we have the following:  For all $\bm{x}'\in\mathbb{R}^n$ with $\mathrm{supp}(\bm{x}')\subset S$ satisfying $\min_{i=0,1}\|\bm{x}'-(-1)^i\bm{x}\|_2\leq \delta_0$ for some sufficiently small constant $\delta_0$,  suppose that  {for     $\beta\asymp \log n$} we calculate $$\tau(\bm{x}')=\big[{\frac{\beta}{m^2}\sum_{i=1}^m\big((\bm{x}')^\top\bm{A}_i\bm{x}'-y_i\big)^2}\big]^{1/2}\cdot\|\bm{x}'\|_2$$ as in (\ref{defi:gradthre}), 
   and then update $\bm{x}'$ to $\bm{x}''$ by $$
       \bm{x}''=\mathcal{T}_{\frac{\mu}{\phi^2}\tau(\bm{x}')}\big(\bm{x}'-\frac{\mu}{\phi^2}\nabla f(\bm{x}')_S\big)$$ as in (\ref{defi:Supdate}), 
   where $\nabla f(\bm{x}')$ is calculated as in (\ref{defi:gradient}) and {$0<\mu<c_0$ for some sufficiently small absolute constant $c_0$}, then it holds that 
   \begin{equation}\begin{aligned}\nonumber
      &\min_{i=0,1} \|\bm{x}''-(-1)^i\bm{x}\|_2 \leq \Big(1-\frac{\mu}{16}\Big)\min_{i=0,1} \|\bm{x}'-(-1)^i\bm{x}\|_2 .\end{aligned}
   \end{equation}
\end{lem}
\begin{proof}
{To preserve the presentation flow of the main body, we defer
the detailed proof of the above Lemma \ref{lem5} to Appendix \ref{app:subsec:proof_lemma_7}.}    
\end{proof}

\subsubsection{Linear Convergence Guarantee for Algorithm \ref{alg2}}
The linear convergence of the TWF-S sequence $\{\bm{\tilde{x}}_t\}$ follows from the aggregation of Lemmas \ref{lem4} and \ref{lem5}.

\begin{lem}\label{lem8}
    {\rm (Linear convergence of the   TWF-S sequence)\textbf{.}}   Assume that the fixed signal $\bm{x}\in \mathbb{S}^{n-1}$ satisfies $\mathrm{supp}(\bm{x})\subset S$. If $m\geq C_1k^2\log n$ for some large enough $C_1$, $\alpha$  in (\ref{defi:Sselect}) and $\beta$ in (\ref{defi:gradthre}) are some suitably chosen constants, {$\mu$ in (\ref{defi:Supdate}) satisfies $0<\mu<c_0$ for some absolute constant $c_0$,} then with probability at least $1-C_2\exp(-\Omega(k))-2n^{-9}$ we have for any positive integer $t$ that
    \begin{equation}\nonumber
        \min_{i=0,1} \|\bm{\tilde{x}}_t-(-1)^i\bm{x}\|_2 \leq \Big(1-\frac{\mu}{16}\Big)^t .
    \end{equation}
\end{lem}
However, the ideal sequence $\{\bm{\tilde{x}}_t\}$ is impractical in that it requires the knowledge of $\mathrm{supp}(\bm{x})$ in both spectral initialization and the subsequent refinement. To establish the linear convergence guarantee for the actual sequence $\{\bm{\hat{x}}_t\}$ produced by the proposed TWF, our strategy is to show $\bm{\tilde{x}}_t=\bm{\hat{x}}_t$ holds for $t=0,1,...,T$ with high probability (w.h.p.). To this end, we again analyze the initialization and iteration separately.
\begin{lem}\label{lem7}{\rm (The initialization coincides w.h.p.)}
Assume that $\alpha$ is some suitably chosen constant, then with probability at least $1-2n^{-9}$ we have $\bm{\hat{x}}_0=\bm{\tilde{x}}_0$.
\end{lem}
\begin{proof}
     It is not hard to see that $\widehat{S}_0=\widetilde{S}_0$ implies $\bm{\hat{x}}_0=\bm{\tilde{x}}_0$, where $\widehat{S}_0$ and $\widetilde{S}_0$ are defined in (\ref{defi:hatS_0}) and (\ref{defi:Sselect}), respectively. Obviously, $\widetilde{S}_0\subset \widehat{S}_0$, so it remains to show that $\widehat{S}_0\subset \widetilde{S}_0$, which amounts to showing
     $$|I_l|\leq \alpha\phi^2\sqrt{\frac{\log n}{m}},~l=k+1,...,n.$$
     From Lemma \ref{lem3}, with probability at least $1-2n^{-9}$ it holds for all $l\in [n]$ that $|I_l-x_l^2|\leq c'\sqrt{\frac{\log n}{m}}$, which implies 
     $$|I_l|\leq c'\sqrt{\frac{\log n}{m}},~l=k+1,...,n.$$
     Therefore, with a suitably chosen constant $\alpha$ such that $\alpha\phi^2>c'$ (which can be guaranteed by $\alpha >c'/0.9^2$ according to Lemma \ref{lem1}), we arrive at $$|I_l|\leq \alpha \phi^2 \sqrt{\frac{\log n}{m}},~l=k+1,...,n,$$ as desired.
\end{proof}
\begin{lem}\label{lem10}
{\rm (The iteration sequence coincides w.h.p.)}
      Given the iteration number $T$, if we choose 
    $\beta = C_1\log n$ for a sufficiently large $C_1$,
    then with probability at least $  1-C_2Tn^{-8}$ we have $\bm{\hat{x}}_t=\bm{\tilde{x}}_t$ for $t=0,1,2,...,T$.
\end{lem}
\begin{proof}
    From Lemma \ref{lem7} we know $\bm{\hat{x}}_0=\bm{\tilde{x}}_0$ holds with probability at least $1-2n^{-9}$. Suppose $\bm{\hat{x}}_0=\bm{\tilde{x}}_0$ holds, then     the desired event $\{\bm{\hat{x}}_t=\bm{\tilde{x}}_t,t=0,1,...,T\}$ is equivalent to  $\{\varphi(\bm{\hat{x}}_t)=\eta(\bm{\tilde{x}}_t):t=0,1,...,T-1\}$, where $\varphi(\bm{\hat{x}}_{t})$ and $\eta(\bm{\tilde{x}}_t)$ are the update rules for TWF and TWF-S, as defined in (\ref{defi:twfupdate}) and (\ref{defi:Supdate}). For any $j\in [n]\setminus S$ and $t=0,1,...,T$ we define the event $$E_{jt}=\big\{\big|\nabla_j f(\bm{\tilde{x}}_t)\big|<\tau(\bm{\tilde{x}}_t)\big\},$$
    where $\nabla_j f(\bm{\hat{x}}_t)$ is the $j$-th entry of $\nabla f(\bm{\hat{x}}_t)$. Building on $\bm{\hat{x}}_0=\bm{\tilde{x}}_0$, we first claim that the event $$\bigcap_{j\in [n]\setminus S}\bigcap_{0\leq t\leq T-1}E_{jt}$$ can imply the desired $$\{\varphi(\bm{\hat{x}}_t)=\eta(\bm{\tilde{x}}_t),t=0,1,...,T-1\}.$$

    \noindent{\textbf{(Proving the claim)}}
    Specifically,  $\bigcap_{j\in [n]\setminus S}E_{j0}$ states that the last $n-k$ entries of $\nabla f(\bm{\tilde{x}}_0)$ are smaller than  $\tau(\bm{\tilde{x}}_0)$ in magnitude; combined with  $\mathrm{supp}(\bm{\tilde{x}}_0)\subset [k]$ we can use $\bm{\hat{x}}_0=\bm{\tilde{x}}_0$ and the property of the thresholding operator to obtain
    \begin{equation}
       \begin{aligned}
            \label{3.31}&\bm{\hat{x}}_1=\varphi(\bm{\hat{x}}_0)=\varphi(\bm{\tilde{x}}_0)=\mathcal{T}_{\frac{\mu}{\phi^2}\tau(\bm{\tilde{x}}_0)}\big(\bm{\tilde{x}}_0-\frac{\mu}{\phi^2}\nabla f(\bm{\tilde{x}}_0)\big)\\
            &=\mathcal{T}_{\frac{\mu}{\phi^2}\tau(\bm{\tilde{x}}_0)}\big(\bm{\tilde{x}}_0-\frac{\mu}{\phi^2}\nabla f(\bm{\tilde{x}}_0)_S\big)=\eta(\bm{\tilde{x}}_0)=\bm{\tilde{x}}_1.
       \end{aligned}
    \end{equation}
          Built on $\bm{\hat{x}}_1=\bm{\tilde{x}}_1$, $\bigcap_{j\in [n]\setminus S}E_{j1}$ allows us to similarly obtain $\bm{\hat{x}}_2=\bm{\tilde{x}}_2$ as in (\ref{3.31}). Repeating such arguments $T$ times justifies that on the high-probability event $\bm{\hat{x}}_0=\bm{\tilde{x}}_0$, $\bigcap_{j\in [n]\setminus S}\bigcap_{0\leq t\leq T-1}E_{jt}$ implies the desired event, as claimed. Therefore, it remains to show $\bigcap_{j\in[n]\setminus S}\bigcap_{0\leq t\leq T-1}E_{jt}$ holds with the promised probability. We begin with a union bound
         \begin{equation}
             \begin{aligned}
                 \label{3.32}
                 &\mathbbm{P}\Big(\bigcap_{j\in [n]\setminus S}\bigcap_{0\leq t\leq T-1}E_{jt}\Big)\\&\geq 1-\mathbbm{P}\Big(\bigcup_{j\in [n]\setminus S}\bigcup_{0\leq t\leq T-1}E_{jt}^c             \Big)\\&\geq 1-\sum_{j=k+1}^n\sum_{t=0}^{T-1}\mathbbm{P}(E_{jt}^c), 
             \end{aligned}
         \end{equation}
          and we want to further upper-bound $\mathbbm{P}(E_{jt}^c)$ where $E_{jt}^c=\big\{|\nabla_jf(\bm{\tilde{x}}_t)|\geq \tau(\bm{\tilde{x}}_t)\big\}$.

          \noindent{\textbf{(Bounding $\mathbbm{P}(E_{jt}^c)$)}}
          We use $[\bm{r}^{(i)}_j]^\top$ and $\bm{c}^{(i)}_j$ to denote the $j$-th row and $j$-th column of $\bm{{A}}_i$ respectively, then the $j$-th row of $\bm{\widetilde{A}}_i$ is $\frac{1}{2}\big([\bm{r}^{(i)}_j]^\top+[\bm{c}^{(i)}_j]^\top\big)$. Thus, from (\ref{defi:gradient}) we have \begin{equation}
              \label{3.33}\nabla_jf(\bm{\tilde{x}}_t)=\frac{1}{2m}\sum_{i=1}^m\big(\bm{\tilde{x}}_t^\top\bm{A}_i\bm{\tilde{x}}_t-y_i\big)\big(\bm{r}^{(i)}_j+\bm{c}^{(i)}_j\big)^\top\bm{\tilde{x}}_t.
          \end{equation}
          As stated in Proposition \ref{pro1} and analyzed in its proof, $\{y_i\}$ and $\{\bm{\tilde{x}}_t\}$ are independent of $\{\bm{r}^{(i)}_j,\bm{c}_j^{(i)}\}$. Besides, since $\mathrm{supp}(\bm{\tilde{x}}_t)\subset S$, $\bm{\tilde{x}}_t^\top\bm{A}_i\bm{\tilde{x}}_t=\bm{\tilde{x}}_t^\top[\bm{A}_i]_{S,S}\bm{\tilde{x}}_t$ is also independent of $\{\bm{r}_j^{(c)},\bm{c}_j^{(i)}\}$. Therefore, we can utilize the randomness of $\{\bm{r}_j^{(c)},\bm{c}_j^{(i)}\}$ while conditioning on everything else. More specifically, because the i.i.d. random vectors $\{\bm{r}^{(1)}_j,...,\bm{r}^{(m)}_j\}$ follow the standard Gaussian distribution $\mathcal{N}(0,\bm{I}_n)$, the randomness of $\{\bm{r}_j^{(i)}\}$ yields 
          \begin{align*}
              &\frac{1}{2m}\sum_{i=1}^m \big(\bm{\tilde{x}}_t^\top\bm{A}_i\bm{\tilde{x}}_t-y_i\big)\big[\bm{r}_j^{(i)}\big]^\top\bm{\tilde{x}}_t\\
              &\sim\mathcal{N}\big(0,\frac{\|\bm{\tilde{x}}_t\|_2^2}{4m^2}\sum_{i=1}^m\big(\bm{\tilde{x}}_t^\top\bm{A}_i\bm{\tilde{x}}_t-y_i\big)^2\big),
          \end{align*} 
          which remains valid if $\bm{r}_j^{(i)}$ is substituted with $\bm{c}_j^{(i)}$. Thus, combined with (\ref{3.33}), it holds with probability at least $1-2\exp(-u^2)$ for some constant $c_2$ and for any $u>0$ that
          \begin{equation}
              \begin{aligned}
                  \label{3.35} |\nabla_jf(\bm{\tilde{x}}_t)|\leq \frac{c_2\|\bm{\tilde{x}}_t\|_2u}{m}\Big[{\sum_{i=1}^m\big(\bm{\tilde{x}}_t^\top\bm{A}_i\bm{\tilde{x}}_t-y_i\big)^2}\Big]^{1/2} 
              \end{aligned}
          \end{equation}
          Recall that $\tau(\bm{\tilde{x}}_t)$ is given in (\ref{defi:gradthre}) with $\beta=C_1\log n$ under a sufficiently large $C_1$. Thus,   setting $u=\frac{\sqrt{C_1\log n}}{c_2}$ in (\ref{3.35}) yields $\mathbbm{P}(E_{jt}^c)\leq 2n^{-9}$.  Substituting this into (\ref{3.32}), the result immediately follows.
\end{proof}
Combining Lemmas \ref{lem8}-\ref{lem10}, we obtain the theoretical guarantee for the proposed TWF. 
\begin{theorem}
    \label{thm1}{\rm (Global linear convergence of TWF)} Suppose that Algorithm \ref{alg2} is provided with some suitably chosen $\alpha\asymp 1$ for (\ref{defi:hatS_0}), $\beta\asymp \log n$ for (\ref{235}),     iteration number $T$,  {step size $\mu$ obeying $0<\mu<c_0$ for some absolute constant $c_0$}, and we let $\bm{\hat{x}}_0,...,\bm{\hat{x}}_T$ be the sequence produced by Algorithm \ref{alg2}. If $m\geq C_0 k^2\log n$ for sufficiently large $C_0$, then  with probability at least $1-C_1Tn^{-8}-C_2\exp(-\Omega(m))$, it holds for any $t\in \{0,1,\ldots,T\}$ that
    \begin{equation}\nonumber
        \min_{i=0,1}\|\bm{\hat{x}}_t- (-1)^i\bm{x}\|_2 \leq \Big(1-\frac{\mu}{16}\Big)^t.
    \end{equation}
\end{theorem}

\section{Solving Generative Quadratic System via Projected Gradient Descent}\label{sec:PGD}
In this section, we consider solving the generative quadratic system where $\bm{x}\in G(\mathbb{B}^k(r))$ for some $L$-Lipschitz generative model $G:\mathbb{B}^k(r)\to \mathbb{R}^n$. Analogously, we will propose a two-step algorithm called projected gradient descent (PGD). In PGD, the first step (see Algorithm~\ref{alg3}) provides a good initialization by the projected power method developed in \cite{liu2022generative}; the second step (see Algorithm~\ref{alg4}) resembles Algorithm \ref{alg2}, with the difference that the thresholding operation $\mathcal{T}_{\frac{\mu}{\phi^2}\tau(\bm{\hat{x}_t})}(\cdot)$ in (\ref{defi:twfupdate}) is replaced by the projection operator $\mathcal{P}_G(\cdot)$ defined as the projection onto $G(\mathbb{B}^k(r))$ under $\ell_2$ norm. This is a result of replacing the sparse prior $\bm{x}\in \Sigma^n_k$ with the generative one $\bm{x}\in G(\mathbb{B}^k(r))$. For technical convenience, we suppose that the generative prior is normalized, i.e., \begin{equation}\label{normalized}\nonumber
    G(\mathbb{B}^k(r))\subset \mathbb{S}^{n-1},
\end{equation} which is also adopted in the relevant parts of \cite{Liu:CPR:2021,liu2022generative}. Actually, our technical analysis does not essentially rely on this normalization assumption, and for a given unnormalized $G(\cdot)$ one can consider its normalized version that is also Lipschitz continuous, see the discussions in \cite[Remark 1]{liu2020sample} and \cite[Page 5]{Liu:CPR:2021}. Moreover, we suppose that we can find some point positively correlated to $\bm{x}$ so that we can exactly recover $\bm{x}$ without trivial ambiguity.

 \begin{algorithm} 
\caption{Projected Power Method}\label{alg3}
\begin{algorithmic}[1]
 \Statex \textbf{Input:} $(\bm{A}_i,y_i)_{i=1}^n$; initial vector $\bw_0$ with $\|\bw_0\|_2=1$

\item Construct a $n\times n$   matrix $ 
    \bm{S}_{in}=\frac{1}{m}\sum_{i=1}^m y_i\bm{A}_i$
and then symmetrize it to
$\bm{\widetilde{S}}_{in} = \frac{1}{2}\big(\bm{S}_{in}+\bm{S}_{in}^\top\big)=\frac{1}{m}\sum_{i=1}^m y_i\bm{\widetilde{A}}_i$. 
\item  Calculate $\hat{\bm{w}}$ as
    \begin{equation}
    \label{defi:ppm}
    \hat{\bm{w}}=\mathcal{P}_G(\bm{\widetilde{S}}_{in}\bm{w}_0).
\end{equation}

\Statex\textbf{Output:} $\hat{\bm{w}}$
\end{algorithmic}
\end{algorithm}

\begin{algorithm} 
\caption{Projected Gradient Descent}\label{alg4}
\begin{algorithmic}[1]
 \Statex \textbf{Input:}   $(\bm{A}_i,y_i)_{i=1}^n$; thresholding operator $\mathcal{T}$;  step size $\mu$;    number of iterations $T$.

 \item Run Algorithm \ref{alg3} to get $\bm{{x}}_0 = \hat{\bm{w}}$. 

 \item \textbf{for} $t=0,1,...,T-1$

Compute $\nabla f(\bm{{x}}_t)$ as in (\ref{defi:gradient}), then update $\bm{x}_t$ to $\bm{x}_{t+1}$ by projected gradient descent 
\begin{equation}    \label{defi:pgdupdate}\bm{{x}}_{t+1}=\mathcal{P}_G\big(\bm{{x}}_t-\mu\cdot\nabla f(\bm{{x}}_t)\big).
\end{equation}

\noindent
\textbf{end}

\Statex\textbf{Output:} $\bm{{x}}_T$.
\end{algorithmic}
\end{algorithm}

Though Algorithms \ref{alg3}-\ref{alg4} involve the projection operator $\mathcal{P}_G(\cdot)$ and thus do not straightforwardly lead to “practical” decoders, we note that such projection operator is commonly adopted in the area of signal reconstruction under generative priors. Similarly to~\cite{peng2020solving, shah2018solving, hyder2019alternating, liu2022generative, liu2022misspecified}, we implicitly assume the exact projection in our theoretical analysis. In practice, approximate methods such as gradient- and GAN-based projections~\cite{shah2018solving, raj2019gan} have been shown to work well.

We will present Theorems \ref{thm:init_gen}-\ref{thm:generefine} as theoretical guarantees for Algorithms \ref{alg3}-\ref{alg4}, respectively. 
Due to the space limit, we only provide a proof sketch here and place the complete proof in the supplementary material.
\subsection{Projected Power Method (Algorithm \ref{alg3})}

\begin{theorem}\label{thm:init_gen}{\rm (Closeness of the projected power method)}
     In Algorithm \ref{alg3}, suppose that   
     \begin{equation} \label{eq:initCond_ppower}
         \bx^\top \bm{w}_0>c_0
     \end{equation}
     for some $c_0>0$  that can be small enough. Then, if $m\geq C_1k\log(Lrn)$ for some sufficiently large $C_1$, with probability at least $1-\exp(-C_2k\log(Lrn))$ we have 
\begin{equation}\label{(30)}
    \|\bm{\hat{w}} -\bx\|_2 \le \frac{C}{c_0}\sqrt{\frac{k \log (nLr)}{m}}. 
\end{equation}
\end{theorem}
\begin{proof}[A Sketch of the proof] We derive the result from Lemma 15 (in the supplementary material) that is adapted from the general framework built in \cite[Lem. 2]{liu2022generative}. Specifically, we can invoke  Lemma 15 with the data matrix $\bm{V}=\bm{\widetilde{S}}_{in}$ that concentrates around its expectation $\bm{xx}^\top$ (Lemma \ref{lem2}) with error $\bm{E}= \bm{\widetilde{S}}_{in}-\bm{xx}^\top$, since the two conditions $\textbf{(C1)},\textbf{(C2)}$ (needed in Lemma 15)  
are verified in Lemma 13 (in the supplementary material).  Then, properly rescaling $\delta$ gives the desired result.  
\end{proof}
\begin{remark}\label{rem:genepp}
    Since a $d$-layer feedforward neural network generative model is typically $L$-Lipschitz with $L = n^{\Theta(d)}$~\cite{bora2017compressed}, we may set $r = n^{\Theta(d)}$ {so that the upper bound in~(\ref{(30)}) is of order $O(\sqrt{\frac{k \log L}{m}})=O(\sqrt{\frac{kd\log n}{m}})$.} 
\end{remark} 
\begin{remark}\label{rem:trivial}
In certain applications, the assumption that the elements of the underlying signal are all non-negative is a natural one, as is the case with image datasets. If this assumption is valid, then the requirement  $\bx^\top\bw_{0}> c_0$ (for some small positive constant $c_0$) might become relatively minor in practice, e.g., we can set $\bw_0=(\frac{1}{\sqrt{n}},\cdots,\frac{1}{\sqrt{n}})^\top$ if the non-negative $\bm{x}$ is smooth (or ``non-peaky'').  In addition, similar assumptions have been made in relevant works including~\cite{liu2022generative, liu2022misspecified}. In this situation, it is not appropriate to assume that $-\bx$ is also contained in $G(\mathbb{B}^k(r))$. In this work we follow similar convention --- we assume we can find a point $\bm{w}_0$ positively correlated to $\bm{x}$ (as in (\ref{eq:initCond_ppower})) to allow for
 an upper bound on $\|\hat{\bw} -\bx\|_2$ (instead of $\min \{\|\hat{\bw} -\bx\|_2, \|\hat{\bw} +\bx\|_2\}$ in the sparse case). 
\end{remark}

\subsection{Projected Gradient Descent (Algorithm \ref{alg4})}
\begin{theorem}\label{thm:generefine}
    {\rm(Linear convergence of PGD)} Let $\delta\in(0,1)$ be given and  sufficiently small. In Algorithm \ref{alg4}, suppose that $m=C_1k\log(\frac{Lrn}{\delta^2})$ with a sufficiently large $C_1$, and  that the step size $\mu\leq 1$ satisfies 
    \begin{equation}\label{eq:initCond_pgd_gen}
        2-(2-7\|\bm{x}_0-\bm{x}\|_2)\mu < 1-2\varepsilon,
    \end{equation}
    for some constant $\varepsilon \in \big(0, \frac{1}{2}\big]$. Then, with probability at least $1-\exp(-C_2k\log\frac{Lrn}{\delta^2})$, PGD produces a sequence that linearly converges to a point with $O(\delta)$ error, with the following inequality holds for any $t$:
    \begin{equation}\label{(33)}
        \|\bm{x}_{t}-\bm{x} \|_2\leq (1-\varepsilon)^t\cdot\|\bm{x}_0-\bm{x}\|_2+ C_3\delta,~\forall~t\geq 0.
    \end{equation}
\end{theorem}

\begin{proof}[A Sketch of the proof] 
To prove (\ref{(33)}), it suffices to show \begin{equation}\nonumber
    \|\bm{x}_{t+1}-\bm{x}\|_2\leq (1-\varepsilon)\cdot\|\bm{x}_t-\bm{x}\|_2 +O(\delta),~\forall~t\geq 0. 
\end{equation}
Let $\bm{h}_t:=\bm{x}_t-\bm{x}$, and we need to bound $\|\bm{h}_{t+1}\|_2$ by $\|\bm{h}_t\|_2$. We decompose $\|\bm{h}_{t+1}\|_2^2$ into several terms in (37). To bound these terms, the main technique is to first bound it over a $\delta$-net $G(M)$ of $G(\mathbb{B}^k(r))$ {\it (to this end, we only need to invoke Lemma 14 (in the supplementary material) with a union bound over $G(M)$)}, and then control the approximation error of the net {\it (to this end, we find $\bm{u},\bm{v}$ such that $\|\bm{x}_t-u\|_2,\|\bm{x}_{t+1}-\bm{v}\|_2\leq \delta$, and then utilize the simple bound $\|\bm{\tilde{A}}_i\|_{op}=O(\sqrt{n})$, see (41) for instance)}. Finally, the statement can be obtained by properly rescaling $\delta$. We comment that unlike the condition $\mu\in (0,c_0)$ in Lemma \ref{lem5} for the sparse case, to ensure (\ref{eq:initCond_pgd_gen})
the step size $\mu$ in the generative case cannot be overly small. This stems from the differences in the analyses of the thresholding operator $\mathcal{T}_{.}(\cdot)$ and the projection operator $\mathcal{P}_G(\cdot)$, and we note that 
in the context of generative priors, projected gradient descent type algorithms often demand a lower bound on the step size (e.g., \cite{shah2018solving, liu2022misspecified}). 
\end{proof}
\begin{remark}\label{rem:gen_error_comparison}
Based on the discussion in Remark~\ref{rem:genepp}, we know that $\delta$ can be as small as $\frac{1}{n^q}$ for a large positive constant $q$ (for example, $q=100$) without affecting the order of the sample complexity. As a result, after a sufficient number of iterations, the projected gradient descent algorithm can achieve a reconstruction error significantly smaller than that of the projected power method.
\end{remark}
\begin{remark}
    Note that the initialization condition~\eqref{eq:initCond_pgd_gen} for the projected gradient descent algorithm is more stringent than the initialization condition~\eqref{eq:initCond_ppower} for the projected power method. More specifically, For~\eqref{eq:initCond_ppower}, since $c_0>0$ can be small enough, $\|\bx-\bw_0\|_2$ can be close to $\sqrt{2}$. However, for~\eqref{eq:initCond_ppower}, $\|\mathbf{x} -\bx_0\|_2$ must be smaller. For example, when $\|\mathbf{x} -\bx_0\|_2 \ge \frac{2}{7}$, the condition in~\eqref{eq:initCond_pgd_gen} cannot hold. This stricter condition can also be seen in the experiments, where, despite using the same initial vector, the reconstructed images of the projected gradient descent algorithm exhibit much worse results compared to the projected power method.
\end{remark}

\section{Experimental Results}\label{sec:exp}
Using the proposed approaches, we conduct numerical experiments to corroborate our theoretical results.
\begin{itemize}
\item \emph{Sparse prior.} The proposed approach comprises a spectral initialization step and a refinement step via the thresholded Wirtinger flow (\texttt{TWF}). We begin by examining how the distance between the spectral initializer and a global optimizer evolves under varying numbers of measurements. Subsequently, we compare \texttt{TWF} with Wirtinger flow (\texttt{WF}) \cite{Huang:Quadratic:2020} and sparse power factorization (\texttt{SPF}) \cite{Lee:SPF:2018}, and demonstrate that it outperforms the other two methods significantly.
\item \emph{Generative prior.} Our proposed method also consists of two steps. First, we use the projected power method, referred to as \texttt{PPower}, to initialize the procedure. It is shown that, under proper initialization, \texttt{PPower} can produce reasonable good reconstructions. Next, we employ the projected gradient descent algorithm, referred to as \texttt{PGD}, to refine the estimator obtained from~\texttt{PPower}. We also compare our approach with that of performing \texttt{PGD} while using the same initial vector as that of \texttt{PPower} to demonstrate the significance of initialization for \texttt{PGD}.
\end{itemize}

\subsection{Sparse Prior}
\subsubsection{Closeness of Spectral Initializer}
\label{subsec:exp_spec_init}
To verify the effectiveness of the proposed spectral initializer $\widehat{\vx}_0$ that has an estimated sparsity support (SI-S), we performed simulation experiments to compute the relative distances between $\widehat{\vx}_0$ and the global optimizer $\vx$ (up to a change of sign). We set the signal's dimensionality $n=500$ and the number of nonzero entries $k=5$. The nonzero entries were generated independently from the uniform distribution in $[-0.5,0.5]$. The random Gaussian measurement matrix $\mA_i$ was generated from the standard normal distribution $\mathcal{N}(0,1)$, and the quadratic measurement $y_i$ was computed via $y_i=\vx^\top\mA_i\vx$. We gradually increased the number of measurements $m$ from $50$ to $500$, and computed $\widehat{\vx}_0$ using Algorithm \ref{alg1} with the parameter $\alpha$ set to $0.5$. We also computed the standard spectral initialization (SI) $\widehat{\vx}_0^\prime$ proposed in \cite{Huang:Quadratic:2020}, which does not have an estimated sparsity support. The two spectral initializers were compared in terms of the relative distance to the global optimizer $\vx$. For example, the relative distance between $\widehat{\vx}_0$ and $\vx$ is calculated as 

\begin{align}
\label{eq:rel_distance}
\frac{\textnormal{dist}(\widehat{\vx}_0,\vx)}{\|\vx\|_2}=\min_{i\in(1,-1)}\frac{\|\widehat{\vx}_0-(-1)^i\cdot\vx\|_2}{\|\vx\|_2}\,.
\end{align}

\begin{figure}[tpb]
    \centering
    \includegraphics[width=0.6\columnwidth]{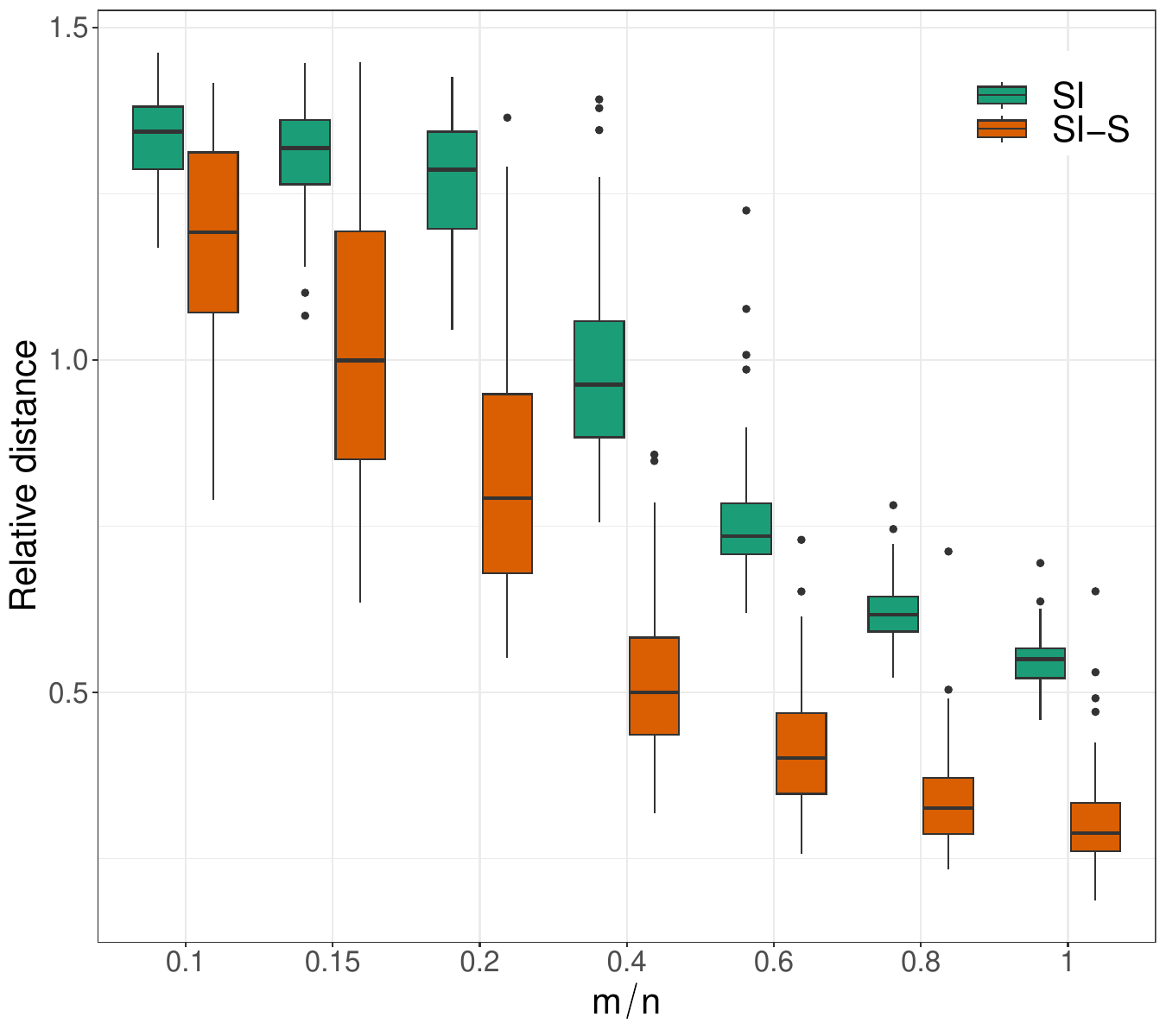}
    \vspace{-1em}
    \caption{Comparison of the standard spectral initialization (SI) and the proposed spectral initialization with estimated sparsity support (SI-S). The relative distances from the initializers to a global optimizer were calculated with $k=5$, $n=500$ and varying number of measurements $m$.}
    \vspace{-1em}
    \label{fig:spec_init_rel_distance}
\end{figure}

For each $m$, we performed $100$ random trials and showed the obtained relative distances using the box plot in Fig. \ref{fig:spec_init_rel_distance}. We can see that the proposed SI-S $\widehat{\vx}_0$ is generally closer to a global optimizer than the standard SI $\widehat{\vx}_0^\prime$. 

\subsubsection{Phase Transition Behavior}
\label{subsec:exp_pt}
We next compare the phase transition curves of \texttt{WF}, \texttt{SPF} and \texttt{TWF} on the reconstruction of sparse signals. We generated the sparse signal $\vx$ and quadratic measurements $\{y_i\}$ following the same process as in section \ref{subsec:exp_spec_init}. We fixed the signal dimensionality $n=100$, and varied the number of nonzero entries $k$ and the number of measurements $m$, where $k\in\{10, 20,\cdots,100\}$ and $m\in{25, 50, \cdots, 200}$. For each combination of $k$ and $m$, we performed $100$ random trials to calculate the success rate. We computed the relative distance between the reconstruction $\hat{\vx}_t$ and the ground-truth $\vx$ as defined in \eqref{eq:rel_distance}. If it was less than $1e^{-3}$, the reconstruction was considered to be a success.

Here we focused on comparing the refinement stage of the three approaches and thus initialized them with the same spectral initializer $\widehat{\vx}_0$. For the proposed \texttt{TWF} summarized in Algorithm \ref{alg2}, we set the threshold tuning parameter $\beta=0.5$, chose the soft-thresholding operation, and adopted an iterative ``damping'' strategy to tune the step size $\mu$. Starting with a relatively large $\mu=0.1$, we reduced $\mu$ by a factor of 2 every 1000 iterations until the maximum number of iteration $T=4000$ was reached. The \texttt{WF} method does not perform iterative thresholding on the solution and is equivalent to the proposed \texttt{TWF} when $\beta$ is simply set to $0$. \texttt{SPF} performs iterative hard thresholding pursuit (\texttt{HTP}) on the estimate and requires a pre-specified sparsity level $k$ for \texttt{HTP} to proceed. In the experiments, we assumed the true sparsity level $k$ was given for \texttt{SPF}.

\begin{figure*}[tbp]
\begin{center}
\vspace{-3em}
\subfloat[]{
\label{fig:wf_ptc}
\includegraphics[height=.29\textwidth]{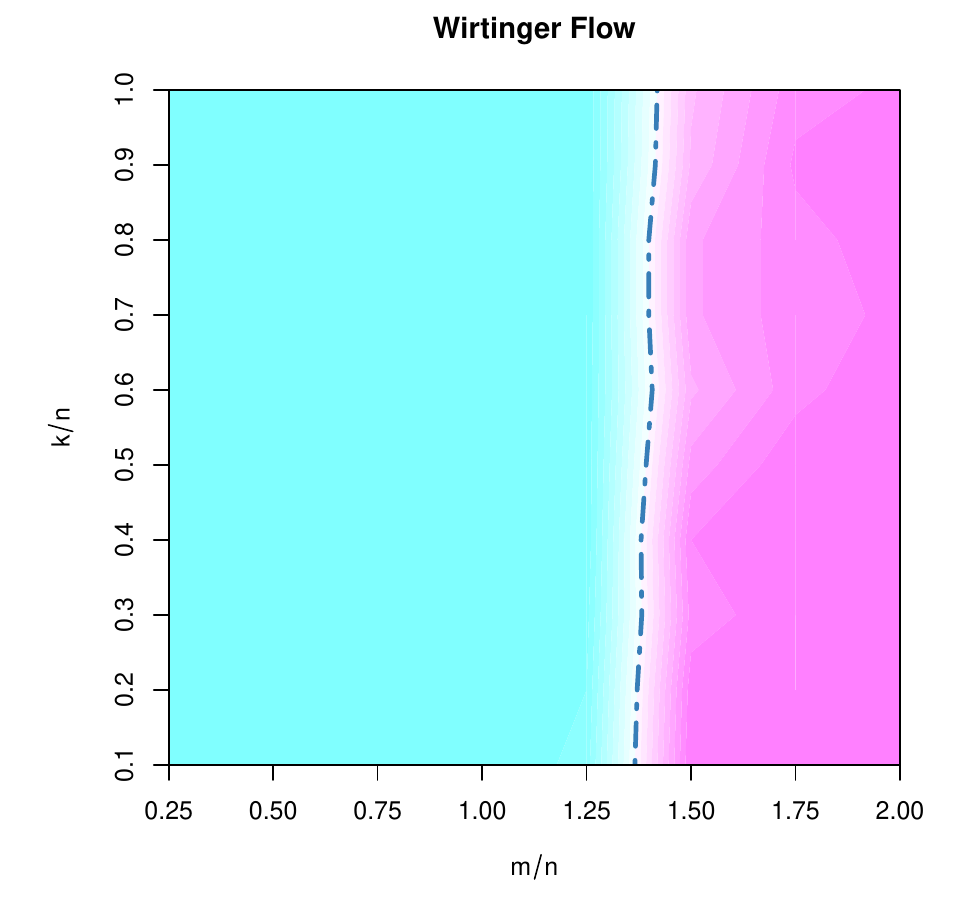}}
\subfloat[]{
\label{fig:spf_ptc}
\includegraphics[height=.29\textwidth]{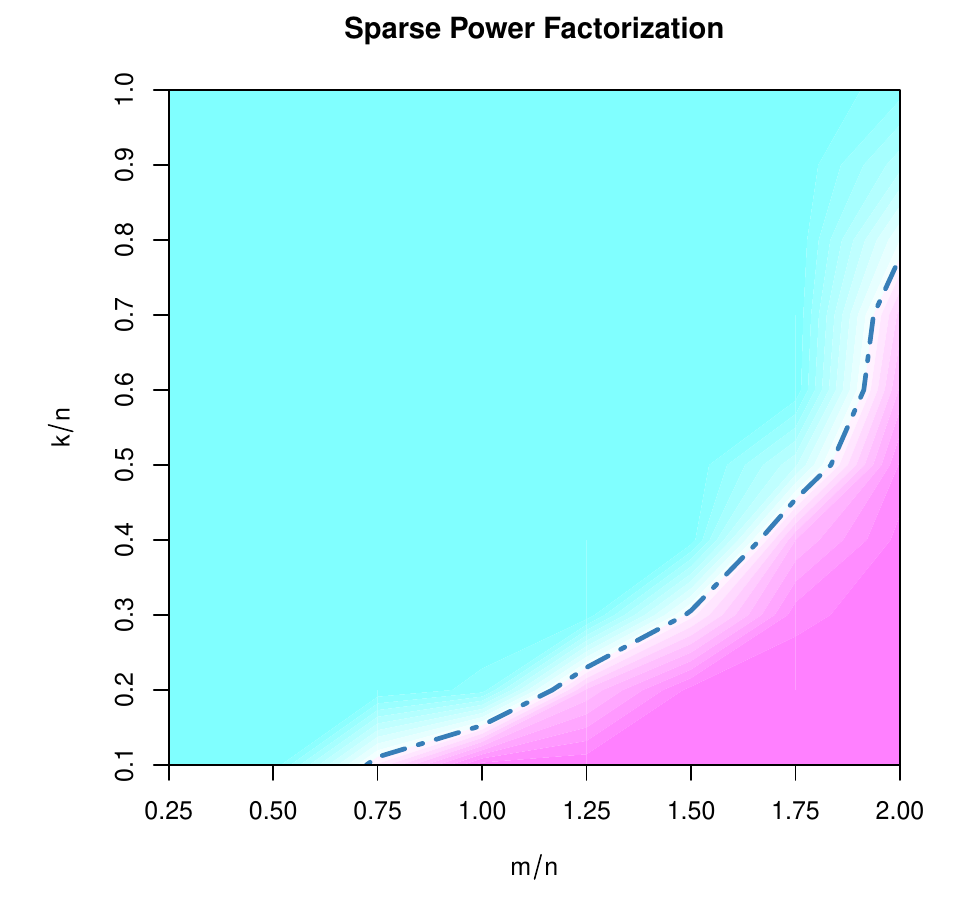}}
\subfloat[]{
\label{fig:twf_ptc}
\includegraphics[height=.29\textwidth]{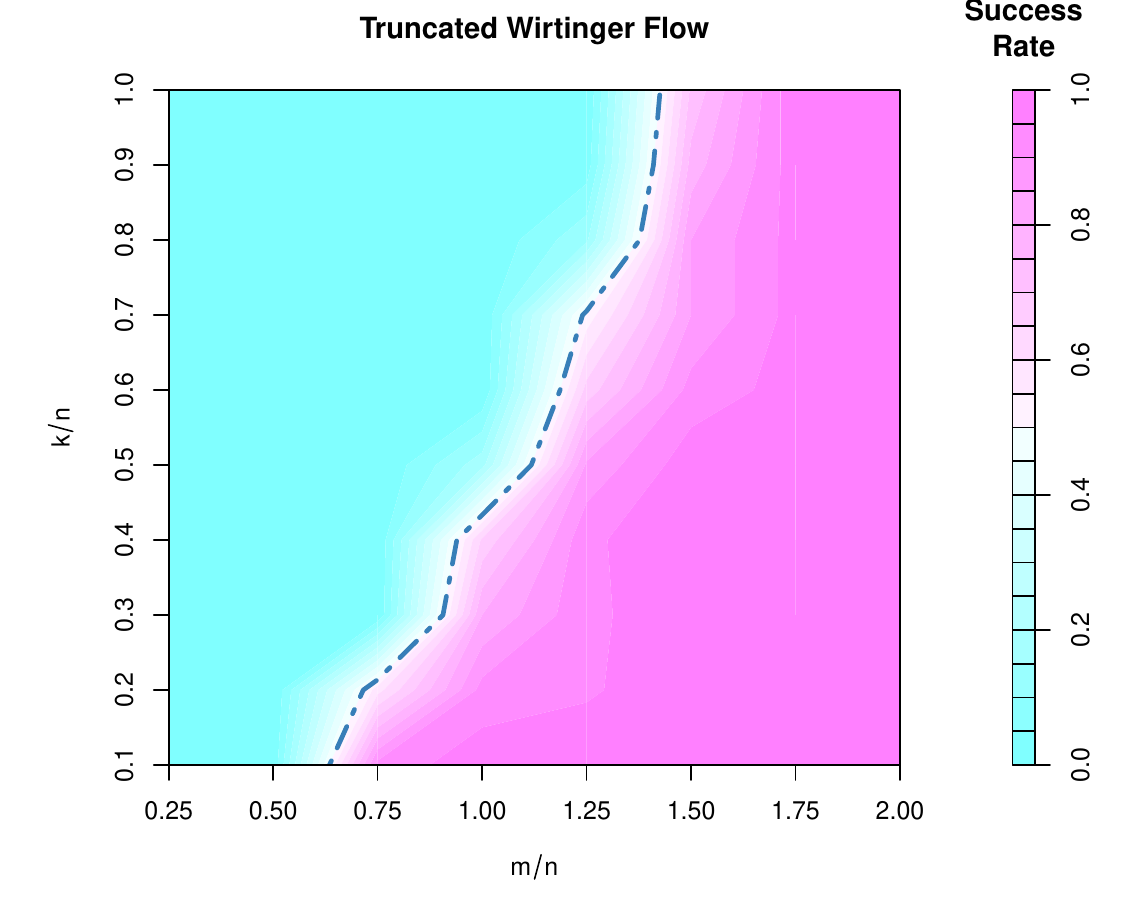}}
\end{center}
\vspace{-1em}
\caption{The success rates of WF, SPF and TWF obtained by varying the sparsity level $k$ and the number of measurements $m$.}
\label{fig:ptc}
\vspace{-1em}
\end{figure*}

The success rates are shown in Fig. \ref{fig:ptc}, with the phase transition curves dividing the plane into the success and failure phases. Fig. \ref{fig:wf_ptc} shows that the performance of \texttt{WF} only depends on the number of measurements $m$, since \texttt{WF} does not take the sparse prior into consideration. Whereas the performance of \texttt{TWF} depends on both $m$ and $k$ as shown in Fig. \ref{fig:twf_ptc}. \texttt{TWF} performs iterative thresholding on the gradient descent update to promote sparse solutions. Exploiting the sparse signal prior allows \texttt{TWF} to achieve better performance than \texttt{WF}. 

By comparing Fig. \ref{fig:spf_ptc} with Fig. \ref{fig:wf_ptc}, we can see that \texttt{SPF} was able to perform better than \texttt{WF} when the sparsity level was low ($k\leq 20$). This is consistent with the global recovery guarantee of \texttt{SPF} for fast-decaying signals in \cite{Lee:SPF:2018}. However, when the sparsity level was high ($k> 20$), \texttt{WF} performed better than \texttt{SPF}. Additionally, by comparing Fig. \ref{fig:spf_ptc} with Fig. \ref{fig:twf_ptc}, we can see that \texttt{SPF} performed much worse than the proposed \texttt{TWF}. Unlike the quadratic model considered in this paper where $y_i=\vx^\top\mA_i\vx$, \texttt{SPF} was proposed for a different model $\widetilde{y}_i=\vv^\top\mA_i\vu$ with two sparse vectors $\vu,\vv$. As a result, when \texttt{SPF} was used to recover $\vx$ from $\{y_i\}$, it is blind to the additional information of $\bm{u}=\bm{v}$ and hence is expected to be sub-optimal in this case. 

\subsubsection{Stable Recovery with Noisy Measurements} {We further tested the robustness of \texttt{TWF} to noise. Specifically, we consider the noisy measurements $y_ i = \bm{x}^\top \bm{A}_i\bm{x} + \zeta_i$, where $\zeta_i$ are independently generated from $\mathcal{N}(0,\sigma^2)$. We note that the norm estimator in Algorithms \ref{alg1}--\ref{alg2}, namely $\phi=(\frac{1}{m}\sum_{i=1}^m y_i^2)^{1/4}$, will no longer be accurate for the following reason: now the expectation of  $y_i^2$ becomes $\mathbbm{E}(y_i^2)=\mathbbm{E}((\bm{x}^\top\bm{A}_i\bm{x})^2)+\mathbbm{E}(\zeta_i^2) = \|\bm{x}\|_2^4+\sigma^2$ that contains the bias $\sigma^2$. Hence, without correcting the bias with the knowledge of $\sigma^2$ (which is unfortunately often  unavailable), our algorithms cannot accurately reconstruct signals with unknown norm.}

However, when the signal norm is known, our numerical results suggest that our algorithms achieve stable signal recovery. In particular, under the above-described noisy measurements, we simulated the reconstruction of $100$-dimensional $3$-sparse signals with unit $\ell_2$ norm, via Algorithms \ref{alg1}--\ref{alg2} with $\phi$ replaced by $1$. We set $\alpha=\beta = 0.5$, $\mu = 0.1$,   executed Algorithm \ref{alg2} with $1000$ iterations, and then averaged the reconstruction error over $50$ independent trials. The results are shown as log-log curves ($\log(\text{error})$ v.s. $\log(m)$) in Figure \ref{fig:noise_twf}, showing that Algorithms \ref{alg1}--\ref{alg2} achieve reasonably small error that decays with $m$ in a rate $O(\frac{1}{\sqrt{m}})$. Interestingly, such error decay rate is consistent with the TWF algorithm for sparse phase retrieval in \cite[Theorem 1]{cai2016optimal}, and we leave the theoretical analysis of this noisy setting to future work.
 
\begin{figure}
 \centering
     \includegraphics[height=0.3\textwidth]{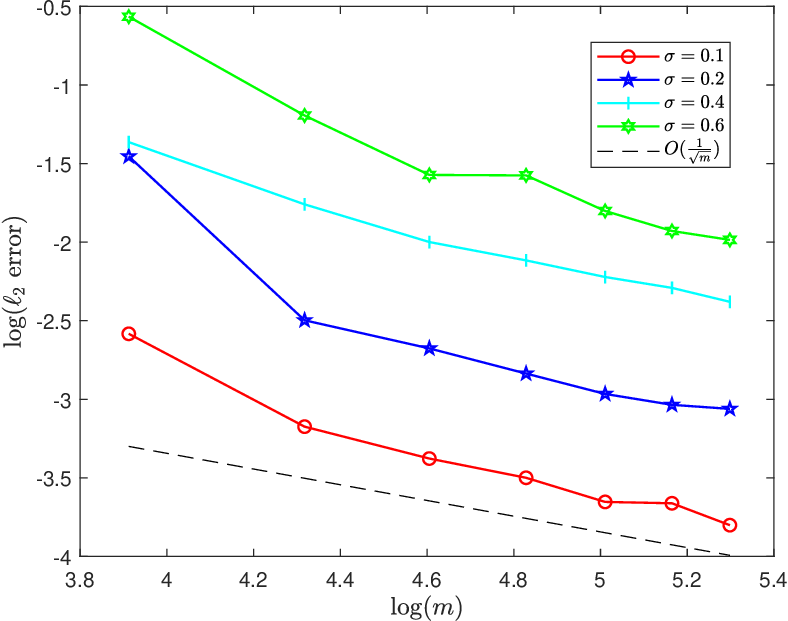}
     \caption{TWF achieves stable recovery of $\bm{x}\in\mathbb{S}^{n-1}$ under additive Gaussian noise.}
     \label{fig:noise_twf}
\end{figure}
\subsection{Generative Prior}

Since the measurement matrices $\bm{A}_i \in \mathbb{R}^{n \times n}$ will be prohibitively large when $n$ is large, we only present numerical results for the MNIST dataset~\cite{lecun1998gradient}, which contains $60,000$ images of handwritten digits, each measuring $28 \times 28$ pixels, resulting in an ambient dimension of $n = 784$. The generative model $G$ for the MNIST dataset is a pre-trained variational autoencoder (VAE) model with a latent dimension of
$k = 20$. The encoder and decoder are both fully connected neural networks with an architecture of $20-500-500-784$. The VAE is trained on the MNIST training set using the Adam optimizer with a mini-batch size of $100$ and a learning rate of $0.001$. To approximate the projection function $\mathcal{P}_{G}(\cdot)$, we follow prior work~\cite{shah2018solving, liu2020sample, peng2020solving, liu2022generative} and employ a gradient descent method utilizing the Adam optimizer with a step size of $100$ and a learning rate of $0.1$. The recovery task is evaluated on a random subset of $10$ images drawn from the testing set of the MNIST dataset.

\subsubsection{Comparing Several Algorithms}
We employ the vector $\bm{o}_1 := \big[\frac{1}{\sqrt{n}}, \frac{1}{\sqrt{n}},\ldots,\frac{1}{\sqrt{n}}\big]^\top \in \mathbb{R}^n$ as the initial vector for the projected power method. We compare the following three approaches: i) the projected power method (denoted by~\texttt{PPower}); ii) the projected gradient descent algorithm using $\bm{o}_1$ as the initial vector (denoted by~\texttt{PGD1}); and iii) the projected gradient descent algorithm using the output of ~\texttt{PPower} as the initial vector (denoted by~\texttt{PGD}). For~\texttt{PGD1} and~\texttt{PGD}, we perform $10$ iterations and fix the step size $\mu$ as $0.9$. To assess the effectiveness of different approaches, we utilize the scale-invariant Cosine Similarity metric, which is defined as $\mathrm{CosSim}(\bm{x}, \hat{\bm{x}}) = \bm{x}^\top\hat{\bm{x}}$, where $\bm{x}$ represents the underlying signal and $\hat{\bm{x}}$ refers to the normalized output vector of each approach. We employ $5$ random restarts to counter the effect of local minima and choose the best result from these restarts. The Cosine Similarity is then averaged over the $10$ test images and these $5$ restarts. All experiments involving generative priors are carried out using Python 3.10.6 and PyTorch 2.0.0 on an NVIDIA RTX 3060 Laptop 6GB GPU. 

We vary the number of measurements $m$ in $\{10, 40, 80, 120, 150, 200, 250, 300, 350, 400\}$ and present the experimental results in Figs.~\ref{fig:MNIST_m120},~\ref{fig:MNIST_m400}, and~\ref{fig:quant_varyM}. The following observations can be gleaned from these figures: i) With the initial vector $\bm{o}_1$,~\texttt{PPower} gives reasonably good reconstructions, but the reconstructed images for~\texttt{PGD1} are much worse. This is consistent with our theoretical findings that suggest the initialization condition~\eqref{eq:initCond_pgd_gen} for the PGD algorithm is more stringent than the initialization condition~\eqref{eq:initCond_ppower} for the projected power method. ii) When using the output of~\texttt{PPower} as the initial vector, the projected gradient descent algorithm~\texttt{PGD} yields the best reconstructions. This observation suggests that the projected gradient descent algorithm has the ability to refine the estimated vector provided by the projected power method, thus improving its reconstruction quality.

\begin{figure}
     \includegraphics[width=\columnwidth]{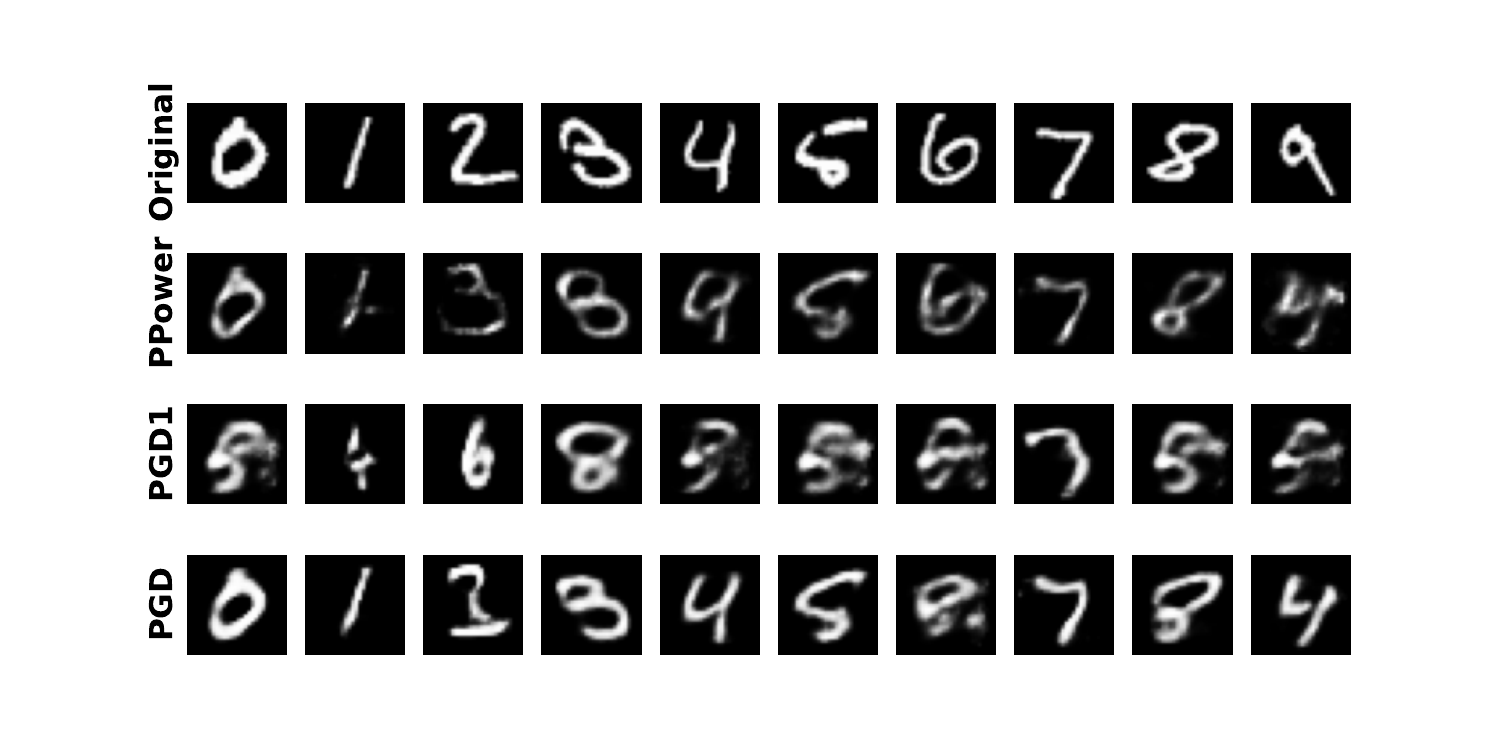}
     \caption{Noiseless reconstructions for MNIST images with $m =120$.}
     \label{fig:MNIST_m120}
\end{figure}

\begin{figure}
 \centering
     \includegraphics[width=\columnwidth]{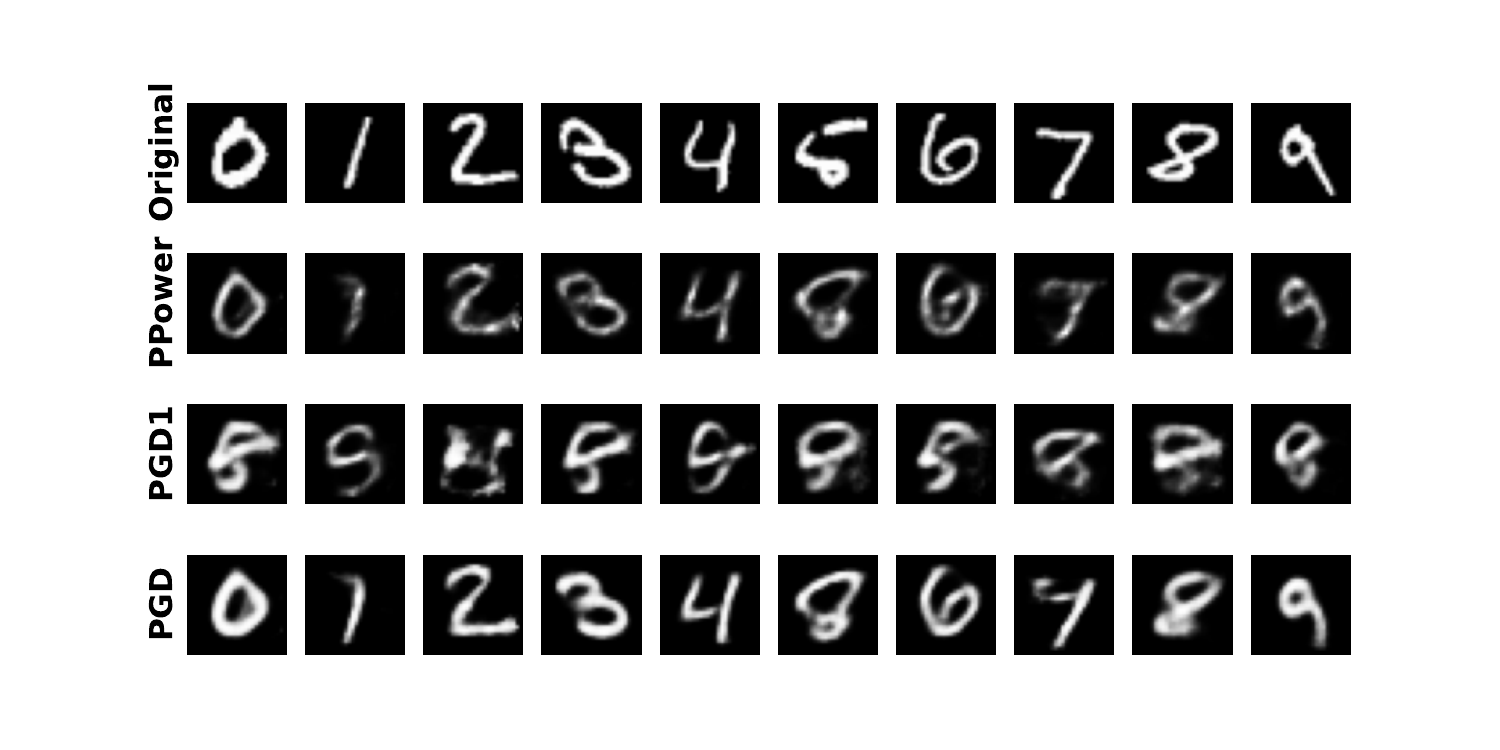}
     \caption{Noiseless reconstructions for MNIST images with $m =400$.}
     \vspace{-1em}
     \label{fig:MNIST_m400}
\end{figure}

 \begin{figure}
 \centering
     \includegraphics[height=0.3\textwidth]{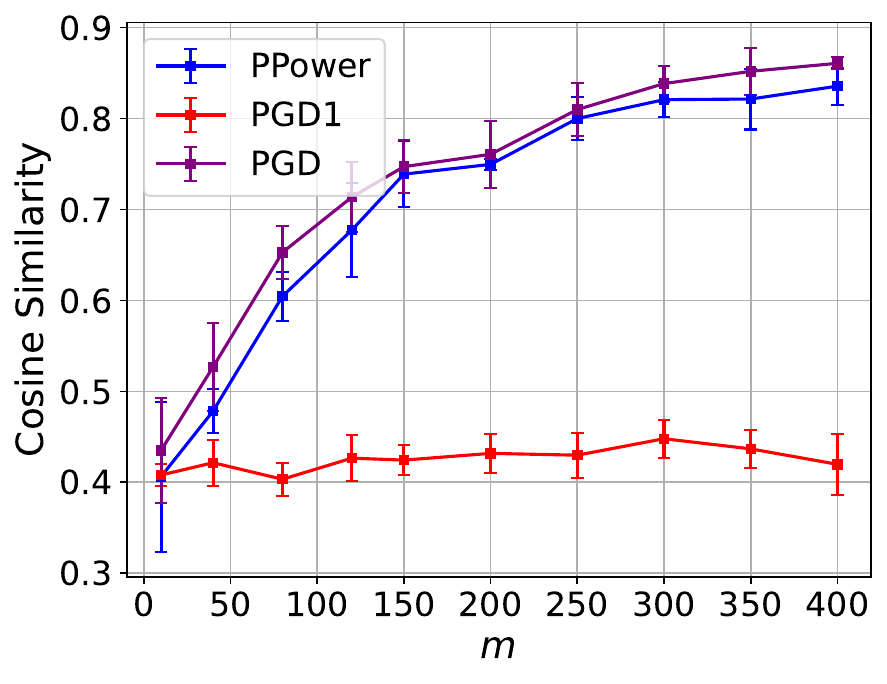}
     \caption{Quantitative results for MNIST with noiseless measurements and varying $m$.}
     \vspace{-1em}
     \label{fig:quant_varyM}
 \end{figure}

\subsubsection{Stable Recovery with Noisy Measurements}
 {We additionally carry out the experiments for the noisy scenario with $y_i = \bx^\top \mA_i \bx + \zeta_i$, where $\zeta_i$ are independently generated from $\mathcal{N}(0,\sigma^2)$ for some $\sigma > 0$. We vary $\sigma$ within the set $\{0.1,0.2,0.5,1,2\}$ and other settings remain the same as those for the noiseless case. The corresponding reconstructed images and quantitative comparison are presented in Figs.~\ref{fig:MNIST_m200_sigma0.2} and~\ref{fig:quant_m200_varySigma} respectively, from which we observe that both \texttt{PPower} and \texttt{PGD} also result in reasonably good reconstructions in the presence of small noise, and \texttt{PGD} still yields the best reconstructions.}

\begin{figure}
\hspace{-0.5cm}
 \centering
     \includegraphics[width=\columnwidth]{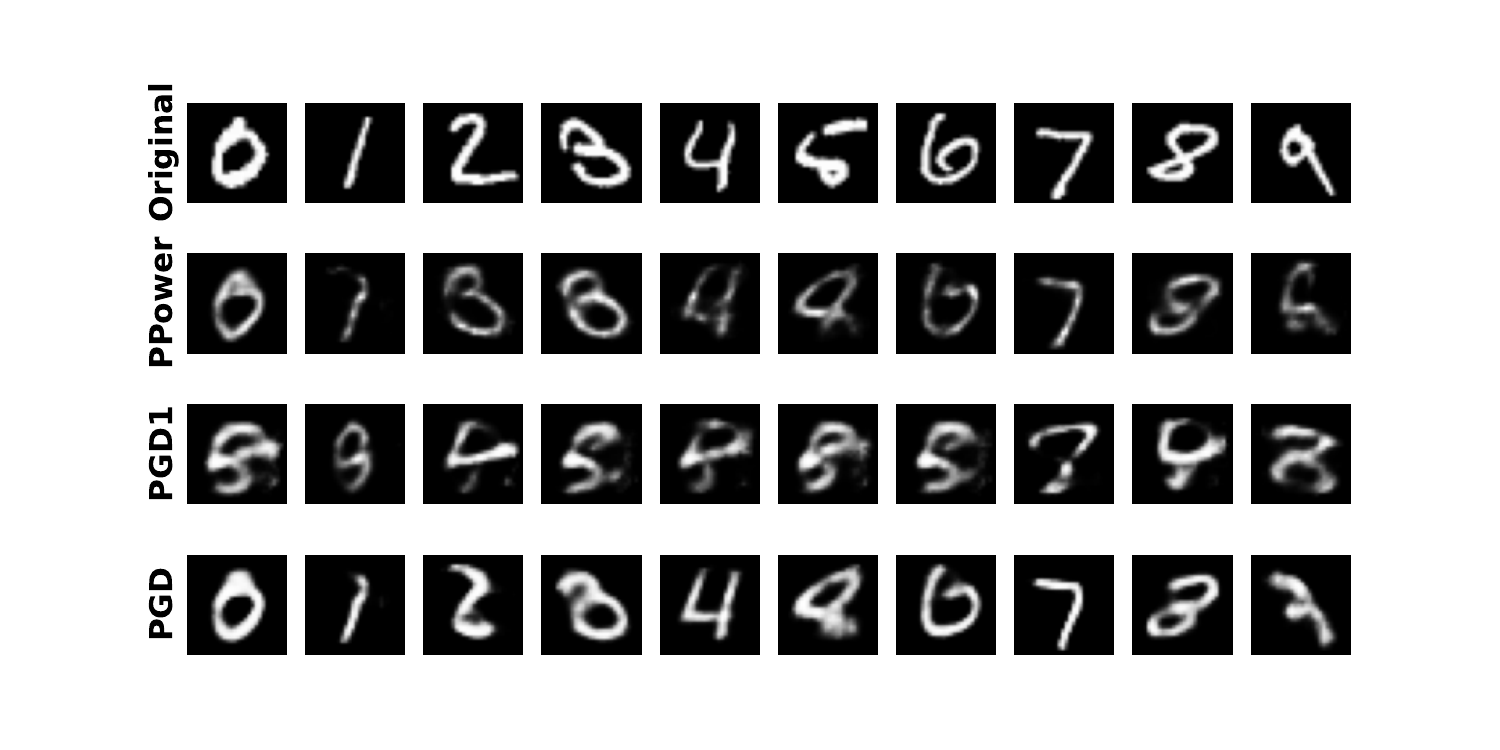}
     \caption{Reconstructed MNIST images for $m =200$ and $\sigma = 0.2$.}
     \label{fig:MNIST_m200_sigma0.2}
\end{figure}

 \begin{figure}
 \centering
     \includegraphics[height=0.3\textwidth]{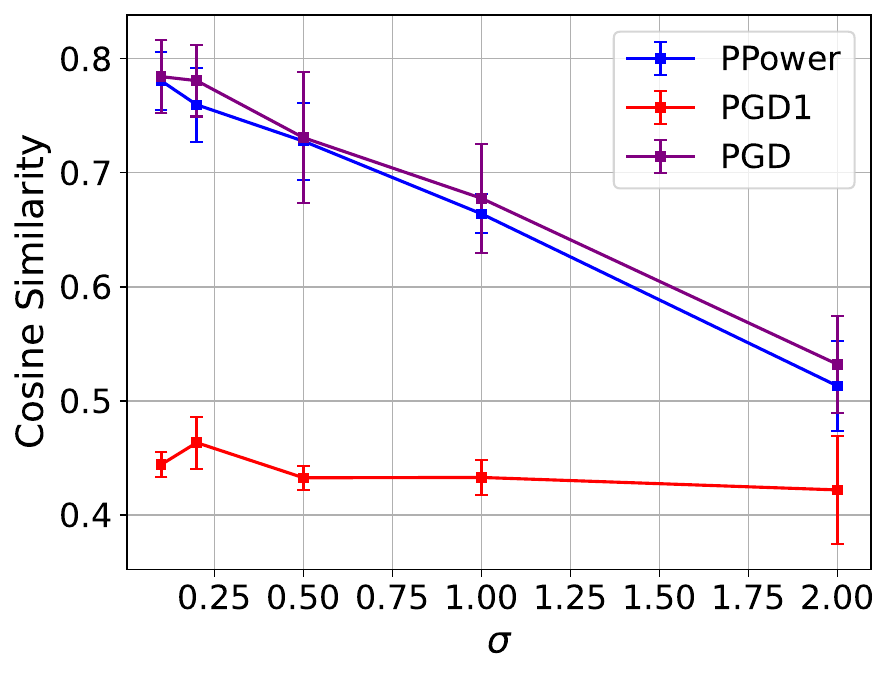}
     \caption{Quantitative results for MNIST with $m=200$ and varying $\sigma$.}
     \vspace{-1em}
     \label{fig:quant_m200_varySigma}
 \end{figure}

\section{Conclusion}\label{sec:conclude}

In this paper, we address the problem of solving a high-dimensional quadratic system $y_i=\vx^\top\mA_i\vx$ with i.i.d. Gaussian matrices $\mA_i$. Our work revolves around two distinct priors on $\bm{x}$ --- a sparse prior and a generative prior. To solve a sparse quadratic system where the signal is $k$-sparse, we propose the thresholded Wirtinger flow (TWF) comprising the two steps of spectral initialization and thresholded gradient descent. TWF requires $m= O(k^2\log n)$ measurements for accurate recovery, where the sub-optimal quadratic dependence on $k$ is incurred in the spectral initialization step.  On the other hand, to solve a generative quadratic system, where $\bm{x}$ lies in the range of a generative model with $k$-dimensional input,  we propose the projected gradient descent (PGD) method. PGD leverages a projected power method for initialization and projected gradient descent for subsequent refinement. Compared to the sparse case, despite that the projection step involved in PGD might not be entirely practical, it only requires a near optimal $m=O(k\log \frac{Lrn}{\delta^2})$ to find a solution $\bm{x}_T$ satisfying $\|\bm{x}_T-\bm{x}\|_2=O(\delta)$. For future research, a notable avenue involves exploring whether the sub-optimal $m=O(k^2\log n)$ in the initialization step of the sparse case can be improved by practical algorithms. Note that this gap is widely known to exist in the literature of sparse phase retrieval, and closing it remains an open question. Moreover, a promising direction is to extend the current theory to the noisy setting where we observe $y_i = \bm{x}^\top \bm{A}_i\bm{x} + \zeta_i$. If the signal norm is known, our experimental results in Figure \ref{fig:noise_twf} suggest that TWF (with $\phi = \|\bm{x}\|_2$) achieves stable reconstruction with error decay rate $O(\frac{1}{\sqrt{m}})$. It is of interest to theoretically justify this error rate and explore its optimality.

\textbf{Acknowledgement:} {The authors would like to thank Dr. Shuai Huang for his insightful discussion and contribution in the literature review and experiments.}

\bibliographystyle{IEEEbib}
\bibliography{libr}

\begin{appendix}
\subsection{Proof of Lemma 7}
\label{app:subsec:proof_lemma_7}
\begin{proof}
    By the definition of the thresholding operator $\mathcal{T}_{\mu  \tau(\bm{x}')/\phi^2}(\cdot)$ we know $\mathrm{supp}(\bm{x}'')\subset S$; and moreover, there exists $\bm{v}\in \mathbb{R}^n$ satisfying $\mathrm{supp}(\bm{v})\subset S$ and $\|\bm{v}\|_{\infty}\leq 1$, such that $$
        \bm{x}''=\bm{x}'-\frac{\mu}{\phi^2}\nabla f(\bm{x}')_S+ \frac{\mu}{\phi^2}\tau(\bm{x}')\bm{v}.$$
    With no loss of generality, we assume $$\min_{i=0,1}\|\bm{x}'-(-1)^i\bm{x}\|_2 =\|\bm{x}'-\bm{x}\|_2\leq \delta_0,$$ then by triangle inequality we obtain 
    \begin{equation}
        \begin{aligned}\nonumber
            \|\bm{x}''-\bm{x}\|_2 \leq \underbrace{\big\|\bm{x}'-\frac{\mu}{\phi^2}\nabla f(\bm{x}')_S -\bm{x}\big\|_2}_{:=I_1} +\underbrace{\big\|\frac{\mu}{\phi^2}\tau(\bm{x}')\bm{v}\big\|_2}_{:=I_2}.
        \end{aligned}
    \end{equation}
    We then proceed to bound the above two terms $I_1,I_2$. We let $\bm{h}=\bm{x}'-\bm{x}$, then by assumption $\mathrm{supp}(\bm{h})\subset S$ and $\|\bm{h}\|_2\leq \delta_0$. 
    
    \noindent{\textbf{(Bounding $I_1$)}}
    We first calculate $I_1^2$ as: \begin{equation}
        \begin{aligned}\nonumber
            I_1^2&=\|\bm{x}'-\bm{x}\|_2^2-\frac{2\mu}{\phi^2}\langle\bm{x}'-\bm{x},\nabla f(\bm{x}')_S\rangle +\frac{\mu^2}{\phi^4}\big\|\nabla f(\bm{x}')_S\big\|_2^2 \\
            &=\|\bm{h}\|_2^2 -\frac{2\mu}{\phi^2}\underbrace{\langle\bm{h},\nabla f(\bm{x}+\bm{h})_S\rangle}_{:=I_{12}}+\frac{\mu^2}{\phi^4}\underbrace{\big\|\nabla f(\bm{x}+\bm{h})_S\big\|_2^2}_{:=I_{13}}
        \end{aligned}
    \end{equation}
\noindent{\textbf{(Bounding $I_{12}$)}} To estimate $I_{12}$, we plug in (\ref{defi:gradient}) and proceed as in  (\ref{hhh}), which provides $I_{12}\geq 2I_{121}+3I_{122}$ if we let
\begin{gather*}
    I_{121}:= \frac{1}{m}\sum_{i=1}^m(\bm{x}^\top\bm{\widetilde{A}}_i\bm{h})^2,\\
    I_{122}= \frac{1}{m}\sum_{i=1}^m (\bm{x}^\top\bm{\widetilde{A}}_i\bm{h})(\bm{h}^\top\bm{\widetilde{A}}_i\bm{h}).
\end{gather*}
\begin{figure*}[!t]
\normalsize
  \begin{equation}
        \begin{aligned}\label{hhh}
            &I_{12}=\Big<\bm{h},\sum_{i=1}^m\frac{((\bm{x}+\bm{h})^\top\bm{A}_i(\bm{x}+\bm{h})-\bm{x}^\top\bm{A}_i\bm{x})\bm{\widetilde{A}}_i({\bm{x}+\bm{h}})}{m}\Big>= \sum_{i=1}^m\frac{(2\bm{x}^\top \bm{\widetilde{A}}_i\bm{h}+\bm{h}^\top\bm{A}_i\bm{h})(\bm{h}^\top\bm{\widetilde{A}}_i\bm{x}+\bm{h}^\top\bm{A}_i\bm{h})}{m}\\&=\frac{1}{m}\sum_{i=1}^m \Big(2\big(\bm{x}^\top\bm{\widetilde{A}}_i\bm{h}\big)^2+3\big(\bm{x}^\top\bm{\widetilde{A}}_i\bm{h}\big)\big(\bm{h}^\top\bm{A}_i\bm{h}\big)+\big(\bm{h}^\top\bm{A}_i\bm{h}\big)^2\Big)\geq\frac{2}{m}\sum_{i=1}^m\big(\bm{x}^\top\bm{\widetilde{A}}_i\bm{h}\big)^2+\frac{3}{m}\sum_{i=1}^m\big(\bm{x}^\top\bm{\widetilde{A}}_i\bm{h}\big)\big(\bm{h}^\top\bm{\widetilde{A}}_i\bm{h}\big)
        \end{aligned}
    \end{equation}
\hrulefill
\vspace*{4pt}
\end{figure*} 
By Lemma \ref{lem12},
       with probability at least $1-2\exp(-k)$ we can proceed on the event (\ref{4.11}), which allows us to lower-bound $I_{121}$ as in (\ref{hhhh}), \begin{figure*}[!t]
\normalsize
   \begin{equation}
        \begin{aligned}\label{hhhh}
            I_{121}= \|\bm{h}\|_2^2\Big(\frac{1}{m}\sum_{i=1}^m\frac{(\bm{x}^\top \bm{\widetilde{A}}_i\bm{h})^2}{\|\bm{h}\|_2^2}-\frac{1}{2}\Big({1+\frac{(\bm{x}^\top\bm{h})^2}{\|\bm{h}\|_2^2}} \Big) +\frac{1}{2}\Big({1+\frac{(\bm{x}^\top\bm{h})^2}{\|\bm{h}\|_2^2}} \Big)\Big)\geq \|\bm{h}\|_2^2\Big(\frac{1}{2}-O\Big(\sqrt{\frac{k}{m}}\Big)\Big)\geq \frac{1}{4}\|\bm{h}\|_2^2,
        \end{aligned}
    \end{equation}
    \begin{equation}
        \begin{aligned}\label{hhhhh}
            |I_{122}|= \|\bm{h}\|_2^3\Big|\frac{1}{m}\sum_{i=1}^m \Big(\bm{x}^\top\bm{\widetilde{A}}_i\frac{\bm{h}}{\|\bm{h}\|_2}\Big)\Big(\frac{\bm{h}^\top}{\|\bm{h}\|_2}\bm{\widetilde{A}}_i\frac{\bm{h}}{\|\bm{h}\|_2}\Big) \frac{\bm{x}^\top\bm{h}}{\|\bm{h}\|_2}+\frac{\bm{x}^\top\bm{h}}{\|\bm{h}\|_2}\Big|
            \leq \|\bm{h}\|_2^3\Big(1+O\Big(\sqrt{\frac{k}{m}}\Big)\Big) \leq 2\|\bm{h}\|_2^3,
        \end{aligned}
    \end{equation} 
    \hrulefill
\vspace*{4pt}
\end{figure*} 
     where the last inequality holds when $m\gtrsim k$. Similarly, we can upper bound $|I_{122}|$ as in (\ref{hhhhh}), 
    where the last inequality holds when $m\gtrsim k$. Combining  the above two bounds   we obtain \begin{equation}\label{3.22}
        I_{12} \geq 2I_{121}-3|I_{122}|\geq \frac{1}{2}\|\bm{h}\|_2^2-6\|\bm{h}\|_2^3 \geq \frac{1}{4}\|\bm{h}\|_2^2,
    \end{equation}
    where the last inequality holds as long as $\|\bm{h}\|_2\leq \delta_0\leq \frac{1}{24}$.

    \noindent{\textbf{(Bounding $I_{13}$)}} By defining the space $\mathcal{X}=  \{\bm{u}\in\mathbb{S}^{n-1}:\mathrm{supp}(\bm{u})\subset S \}$ as in (\ref{defi:calX}), we begin with (\ref{3.23}) with factor $\mathcal{F}$ given by $$ \mathcal{F}:=\sup_{\bm{p},\bm{q}\in \mathcal{X}}\sup_{\bm{u},\bm{v}\in \mathcal{X}}\big|\frac{1}{m}\sum_{i=1}^m \big(\bm{p}^\top\bm{\widetilde{A}}_i\bm{q}\big)\big(\bm{u}^\top\bm{\widetilde{A}}_i\bm{v}\big)\big|.$$ 
    Then, note that $\|\bm{x}+\bm{h}\|_2\leq \|\bm{x}\|_2+\|\bm{h}\|_2\leq 1+\delta_0$, and by the triangle inequality and (\ref{4.11}) (note that we assume we are on this event by removing some probability), it is not hard to see   $\mathcal{F}\leq 2$. Therefore, when $\delta_0$ is sufficiently small, we obtain $\sqrt{I_{13}}\leq 6\|\bm{h}\|_2$, i.e., $I_{13}\leq 36\|\bm{h}\|^2_2$. 
     \begin{figure*}[!t]
\normalsize
    \begin{equation}
        \begin{aligned}\label{3.23}
           \sqrt{I_{13}}&=\big\|\nabla f(\bm{x}+\bm{h}) _S\big\|_2 = \Big\|\frac{1}{m}\sum_{i=1}^m\Big((\bm{x}+\bm{h})^\top\bm{A}_i(\bm{x}+\bm{h}) - \bm{x}^\top\bm{A}_i\bm{x}\Big) [\bm{\widetilde{A}}_i]_{S,S}(\bm{x}+\bm{h})\Big\|_2\\
            &=\sup_{\bm{u}\in\mathcal{X}}\frac{1}{m}\sum_{i=1}^m\Big(2\bm{x}^\top\bm{\widetilde{A}}_i\bm{h}+\bm{h}^\top\bm{\widetilde{A}}_i\bm{h}\Big)\Big(\bm{u}^\top\bm{\widetilde{A}}_i(\bm{x}+\bm{h})\Big) \leq \Big(2\|\bm{h}\|_2\|\bm{x}+\bm{h}\|_2+ \|\bm{h}\|_2^2\|\bm{x}+\bm{h}\|_2\Big)\cdot\mathcal{F}.
        \end{aligned}
    \end{equation} 
\hrulefill
\vspace*{4pt}
\end{figure*} We combine this bound and (\ref{3.22}) to obtain 
   $$
            I_1^2\leq \big(1-\frac{\mu}{2\phi^2}+\frac{36\mu^2}{\phi^4}\big)\|\bm{h}\|_2^2\leq \big(1-\frac{\mu}{4}\big)\|\bm{h}\|_2^2,$$
    where the last inequality is due to $0<\mu<\mu_0$ for some sufficiently small absolute constant $c_0$, and we can suppose that $|\phi-1|$ is also small enough by Lemma \ref{lem1}. Therefore, we arrive at \begin{equation}\label{3.25}
        I_1\leq \sqrt{1-\frac{\mu}{4}}\cdot\|\bm{h}\|_2 \leq \Big(1-\frac{\mu}{8}\Big)\|\bm{h}\|_2. 
    \end{equation}

    \noindent{\textbf{(Bounding $I_2$)}} It remains to bound $I_2$, for which we begin with (\ref{3.26}),
       \begin{figure*}[!t]
\normalsize
 \begin{equation}\begin{aligned}\label{3.26}
        I_2&\leq 2\mu \big|\tau(\bm{x}')\big|\|\bm{v}\|_2\leq 2\mu\sqrt{\frac{k}{m}}\cdot\Big[{\frac{\beta}{m}\sum_{i=1}^m \Big(\big(\bm{x}+\bm{h}\big)^\top\bm{A}_i\big(\bm{x}+\bm{h}\big)-\bm{x}^\top\bm{A}_i\bm{x}\Big)^2}\Big]^{1/2}\\&\leq 2\mu\sqrt{\frac{\beta k}{m}}\Big[\frac{1}{m}\sum_{i=1}^m\Big(4\big(\bm{x}^\top\bm{\widetilde{A}}_i\bm{h}\big)^2+ \big(\bm{h}^\top\bm{\widetilde{A}}_i\bm{h}\big)^2+2\big(\bm{x}^\top\bm{\widetilde{A}}_i\bm{h}\big)\big(\bm{h}^\top\bm{\widetilde{A}}_i\bm{h}\big)\Big)\Big]^{1/2}\\
        &\leq 2\mu\sqrt{\frac{\beta k}{m}}\Big[{\big(4\|\bm{h}\|_2^2+\|\bm{h}\|_2^4+2\|\bm{h}\|_2^3\big)\Big(1+O\Big(\sqrt{\frac{k}{m}}\Big)\Big)}\Big]^{1/2}\leq \frac{{\mu}}{16}\|\bm{h}\|_2,
        \end{aligned}
    \end{equation}
  
\hrulefill
\vspace*{4pt}
\end{figure*} 
    where the last line follows from (\ref{4.11}) as before, and the last inequality can be ensured by $m\gtrsim k\log n$ (recall that $\beta \asymp\log n$). 
    
    \noindent
    {\textbf{(Conclusion)}}
    By substituting (\ref{3.25}) and (\ref{3.26}) into $\|\bm{x}''-\bm{x}\|_2\leq I_1+I_2$, we obtain $$
        \|\bm{x}''-\bm{x}\|_2\leq \big(1-\frac{\mu}{16}\big)\|\bm{h}\|_2=\big(1-\frac{\mu}{16}\big)\|\bm{x}'-\bm{x}\|_2,
        $$
        as desired. 
\end{proof}

    \subsection{Technical Lemmas for Sparse Case}
    Recall that in our analysis we assume $\supp(\bm{x})\subset S=[k]$ with no loss of generality, now we define
\begin{equation}\label{defi:calX}
        \mathcal{X}=\{\bm{u}\in\mathbb{S}^{n-1}:\mathrm{supp}(\bm{u})\subset S \}.
    \end{equation}
    The following two lemmas prove very useful in our analysis of TWF. Their proofs are standard covering arguments and relegated to the supplementary material due to the space limit.
    \begin{lem}\label{lem6}
    {\rm (Concentration of $[\bm{S}_{in}]_{S,S}$, $[\bm{\widetilde{S}}_{in}]_{S,S}$ under operator norm)} Recall that $[\bm{S}_{in}]_{S,S}=\frac{1}{m}\sum_{i=1}^m y_i[\bm{A}_i]_{S,S}$ and $[\bm{\widetilde{S}}_{in}]_{S,S}=\frac{1}{2}\big([\bm{S}_{in}]_{S,S}+[\bm{S}_{in}^\top]_{S,S}\big)$.  Then, with probability at least $1-2\exp(-k)$ we have 
  $$  
        \|[\bm{\widetilde{S}}_{in}]_{S,S}-\bm{xx}^\top\|_{op}\lesssim \sqrt{\frac{k}{m}}+\frac{k}{m}.
   $$  
    In particular, when $m\gtrsim k$, it holds with probability at least $1-2\exp(-k)$ that \begin{equation}\label{4.1}
        \big\|[\bm{\widetilde{S}}_{in}]_{S,S}-\bm{xx}^\top\big\|_{op}\lesssim \sqrt{\frac{k}{m}}.
    \end{equation} 
\end{lem}

We note that $y_i=\bm{x}^\top\bm{A}_i\bm{x}$ for some fixed $\bm{x}\in \mathcal{X}$ and hence the proof of Lemma \ref{lem6} essentially bounds \begin{equation}\label{251}
    \sup_{\bm{u},\bm{v}\in \mathcal{X}}\Big|\frac{1}{m}\sum_{i=1}^m \big(\bm{x}^\top\bm{A}_i\bm{x}\big)\big(\bm{u}^\top \bm{A}_i\bm{v}\big)-(\bm{u}^\top\bm{x})(\bm{v}^\top\bm{x})\Big|.
\end{equation}
Now, we further strengthen Lemma \ref{lem6} to a bound on the quantity that replaces $\bm{x}$ in (\ref{251}) by  (uniformly) all vectors in $\mathcal{X}$ (see (\ref{4.10})), and a quantity that further allows for substituting $\bm{A}_i$ with $\bm{\widetilde{A}}_i$ (see (\ref{4.11})).

\begin{lem}\label{lem12}{\rm (Concentration of $\frac{1}{m}\sum_i (\bm{p}^\top\bm{B}_i\bm{q})(\bm{u}^\top\bm{B}_i\bm{v})$, $\bm{B}_i=\bm{A}_i/\bm{\widetilde{A}}_i$ over $(\bm{p},\bm{q},\bm{u},\bm{v})\in \mathcal{X}^4$)}
    Assume $\bm{A}_i\sim \mathcal{N}^{n\times n}(0,1)$, recall that $\bm{\widetilde{A}}_i = \frac{1}{2}(\bm{A}_i+\bm{A}_i^\top)$, $\mathcal{X}$ is defined  in (\ref{defi:calX}). If $m\gtrsim   k$, then with probability at least $1-2\exp(-k)$ we have (\ref{4.10}) and (\ref{4.11}). 
\begin{align}
    \label{4.10}
    \begin{split}
       & \sup_{\bm{p},\bm{q}\in \mathcal{X}}\sup_{\bm{u},\bm{v}\in \mathcal{X}} \Big|\frac{1}{m}\sum_{i=1}^m\big(\bm{p}^\top\bm{A}_i\bm{q}\big)\big(\bm{u}^\top\bm{A}_i\bm{v}\big)\\
        &~~~~~~~~~~~~~~~~~~\vphantom{\frac{1}{m}\sum_{i=1}^m} -(\bm{p}^\top\bm{u})(\bm{q}^\top\bm{v})\Big|\lesssim \sqrt{\frac{k}{m}};
    \end{split}\\ 
    \label{4.11}
    \begin{split}
    &\sup_{\bm{p},\bm{q}\in \mathcal{X}}\sup_{\bm{u},\bm{v}\in \mathcal{X}}  \Big|\frac{1}{m}\sum_{i=1}^m\big(\bm{p}^\top\bm{\widetilde{A}}_i\bm{q}\big)\big(\bm{u}^\top\bm{\widetilde{A}}_i\bm{v}\big)\\
    &\quad\vphantom{\frac{1}{m}\sum_{i=1}^m} -\frac{1}{2}\Big[(\bm{p}^\top\bm{u})(\bm{q}^\top\bm{v})+(\bm{p}^\top\bm{v})(\bm{q}^\top\bm{u})\Big]\Big|\lesssim\sqrt{\frac{k}{m}}. 
    \end{split}
\end{align}
\end{lem}

\end{appendix}

\clearpage
\newpage

\onecolumn
\begin{center}
{\Huge Supplemental to: Solving Quadratic System with\\ Full-Rank Matrices Using Sparse \\ or Generative Prior}
\end{center}
\vspace{0.5em}
\begin{center}
Junren Chen,~Shuai Huang,~Michael K. Ng,~\IEEEmembership{Senior Member},~Zhaoqiang Liu
\end{center}
\vspace{3em}

\setcounter{subsection}{0}

\begin{multicols*}{2}

\subsection{Technical Lemmas for Generative Case}
\setcounter{lemma}{12}
\setcounter{equation}{34}

\begin{lemma}
    \label{genelem1}
    {\rm (Concentration of $\bm{\widetilde{S}}_{in}$ over finite sets, and a bound on $\|\bm{\widetilde{S}}_{in}-\bm{x}\bm{x}^\top\|_{op}$)} Let $\mathcal{S}_1,\mathcal{S}_2\subset \mathbb{R}^n$ be finite sets, and recall that $\bm{S}_{in}=\frac{1}{m}\sum_{i=1}^m y_i\bm{A}_{i},\bm{\widetilde{S}}_{in}=\frac{1}{2}(\bm{S}_{in}+\bm{S}_{in}^\top)$. If $m\gtrsim \log(|\mathcal{S}_1|\cdot|\mathcal{S}_2|)$, then with probability at least $1-2\exp(-\Omega(\log |\mathcal{S}_1|+\log|\mathcal{S}_2|))$ we have \begin{equation}\begin{aligned}\label{lem131}
       | \bm{s}_1^\top &\bm{\widetilde{S}}_{in}\bm{s}_2-\bm{s}_1^\top\bm{x}\bm{x}^\top\bm{s}_2| \lesssim \sqrt{\frac{\log(|\mathcal{S}_1|\cdot|\mathcal{S}_2|)}{m}}\|\bm{s}_1\|_2\|\bm{s}_2\| _2,\\&~~~~~~~~~~~~~~~~~~~~~~~~~~~~~~~~\forall~\bm{s}_1\in \mathcal{S}_1,\bm{s}_2\in \mathcal{S}_2.
       \end{aligned}
    \end{equation}
    Moreover, if $m=O(n)$, then with probability at least $1-2\exp(-n)$ we have $\big\|\bm{\widetilde{S}}_{in}-\bm{x}\bm{x}^\top\big\|_{op}\lesssim \frac{n}{m}.$
\end{lemma}
\begin{lemma}
    \label{genelem2}
    {\rm (Concentration of $\frac{1}{m}\sum_{i}(\bm{p}^\top\bm{\widetilde{A}}_i\bm{q})(\bm{u}^\top\bm{\widetilde{A}}_i\bm{v})$ for fixed $(\bm{p},\bm{q},\bm{u},\bm{v})$)} Given any fixed $\bm{p},\bm{q},\bm{u},\bm{v}\in \mathbb{R}^n$ and $\epsilon\in(0,1)$, with probability at least $1-2\exp(-cm\epsilon^2)$ we have \begin{equation}\begin{aligned}\nonumber
     & \Big|  \frac{1}{m}\sum_{i=1}^m (\bm{p}^\top\bm{\widetilde{A}}_i\bm{q})(\bm{u}^\top\bm{\widetilde{A}}_i\bm{v})-\frac{1}{2} (\bm{p}^\top\bm{u})(\bm{q}^\top\bm{v})\\
      &~~~~~-\frac{1}{2}(\bm{p}^\top\bm{v})(\bm{q}^\top\bm{u})\Big|\leq C\epsilon\cdot  \|\bm{p}\|_2\|\bm{q}\|_2\|\bm{u}\|_2\|\bm{v}\|_2.
      \end{aligned}
    \end{equation}
\end{lemma}
\subsection{Proofs of the main results for generative case (Theorems 2-3)}

\subsubsection{\textbf{Proof of Theorem 2}} We pause to present a lemma that serves as  the most important ingredient in this proof.

\begin{lemma}
{\em \hspace{1sp}(Adapted from \cite[Lemma~2]{liu2022generative})}\label{lem:ppower_gpca}
    Suppose that $\bm{V} = \lambda \bm{v}\bm{v}^T + \bm{E} \in \mathbb{R}^{n \times n}$ with $\lambda >0$ being a positive constant, $\bm{v} \in G(\mathbb{B}^k(r))$ for some $L$-Lipschitz $G$ mapping from $\mathbb{B}^k(r)$ to $\mathbb{S}^{n-1}$, and $\bm{E}$ satisfying the following two  conditions: 
    \begin{itemize}
        \item \textbf{(C1)} For any given finite sets $S_1, S_2$, if $m = \Omega(\log(|S_1|\cdot |S_2|))$, we have for all $\bm{s}_1 \in S_1$ and $\bm{s}_2 \in S_2$, $|\bm{s}_1^\top\bm{E}\bm{s}_2| \le C\sqrt{\frac{\log(|S_1|\cdot |S_2|)}{m}} \cdot \|\bm{s}_1\|_2 \cdot \|\bm{s}_2\|_2$ {\color{black}with probability at least} $1-2\exp(-\Omega(\log |S_1|+\log|S_1|))$ 

        \vspace{2mm}
        \item \textbf{(C2)} $\|\bm{E}\|_{op}= O(n/m)$.
    \end{itemize}
      Then, for any $\delta >0$, {\color{black} with probability at least $1-2\exp(-\Omega(k\log\frac{Lr}{\delta}))$}, it holds for all $\bm{s} \in G(\mathbb{B}^k(r))$ satisfying $\bm{s}^T\bm{v}> 0$ that
    \begin{equation}\nonumber
        \left\|\mathcal{P}_G(\bm{V}\bm{s}) - \bm{v}\right\|_2 \leq \frac{C}{\bm{s}^\top\bm{v}}\Big(\sqrt{\frac{k \log \frac{Lr}{\delta}}{m}} + \sqrt{\frac{n\delta}{m}}\Big).
    \end{equation}
\end{lemma}

We use Lemma \ref{lem:ppower_gpca} with $\bm{V}=\bm{\widetilde{S}}_{in}$, $\lambda=1$, $\bm{v}=\bm{x}$, $\bm{E}=\bm{\widetilde{S}}_{in}-\bm{xx}^\top$. Then, the two conditions \textbf{C1},\textbf{C2} required in Lemma \ref{lem:ppower_gpca} have been verified in Lemma 13. Because we assume $\bm{x}^\top\bm{w}_0>c_0$, and note that $\bm{\hat{w}}=\mathcal{P}_G(\bm{\widetilde{S}}_{in}\bm{w}_0)$, Lemma \ref{lem:ppower_gpca} gives  
\begin{equation}\nonumber
    \|\bm{\hat{w}}-\bm{x}\|_2 \leq \frac{C}{c_0} \Big(\sqrt{\frac{k\log \frac{Lr}{\delta}}{m}}+\sqrt{\frac{n\delta}{m}}\Big)
\end{equation}
with probability at least $1-2\exp(-\Omega(k\log\frac{Lr}{\delta}))$. Setting $\delta=\frac{1}{n}$ completes the proof. \hfill $\square$

\subsubsection{\textbf{Proof of Theorem 3}} We first consider a given $\delta>0$ satisfying $\delta = O(n^{-1})$ and finally rescale it to the proper scaling as in the theorem.
Let  
    \begin{equation}\nonumber
        \tilde{\bx}_{t+1} = \bx_t - \frac{\mu}{m}\sum_{i=1}^m (\bx_t^\top\tilde{\bA}_i\bx_t - y_i)\tilde{\bA}_i\bx_t.
    \end{equation}
    Since $\bx_{t+1} = \calP_G\big(\tilde{\bx}_{t+1}\big)$ and $\bx \in \calR(G)$, we obtain
    \begin{equation}\nonumber
        \|\tilde{\bx}_{t+1} - \bx_{t+1}\|_2^2 \le \|\tilde{\bx}_{t+1} - \bx\|_2^2.
    \end{equation}
    Then, if letting $\bh_t = \bx_t - \bx$, we obtain
    \begin{align}\nonumber
        \begin{split}
        & \|\bh_{t+1}\|_2^2 = \|\bx_{t+1}-\bx\|_2^2 \le 2\langle\tilde{\bx}_{t+1} -\bx,\bx_{t+1}-\bx \rangle \\
        &  = 2\Big\langle \bh_t -\frac{\mu}{m}\sum_{i=1}^m (\bx_t^\top\tilde{\bA}_i\bx_t-y_i)\tilde{\bA}_i\bx_t, \bh_{t+1}\Big\rangle 
        \end{split}\\
        & = 2\Big\langle \bh_t -\frac{\mu}{m}\sum_{i=1}^m \left(\bh_t^\top\tilde{\bA}_i\bh_t + 2\bh_t^\top \tilde{\bA}_i\bx\right)\tilde{\bA}_i\bx_t, \bh_{t+1}\Big\rangle \label{eq:main_control_gen_0} \\
        \begin{split}
        & = 2\bh_t^\top\bh_{t+1} - \frac{2\mu}{m}\sum_{i=1}^m (\bh_t^\top\tilde{\bA}_i\bh_t) (\bh_{t+1}^\top \tilde{\bA}_i\bx_t) \\
        &- \frac{4\mu}{m}\sum_{i=1}^m (\bh_t^\top\tilde{\bA}_i\bx) (\bh_{t+1}^\top \tilde{\bA}_i\bh_t)  \nn -  \frac{4\mu}{m}\sum_{i=1}^m (\bh_t^\top\tilde{\bA}_i\bx) (\bh_{t+1}^\top \tilde{\bA}_i\bx) 
        \end{split}\\
        & \le 2\bh_t^\top\bh_{t+1} + 2\mu \cdot\Big|\frac{1}{m}\sum_{i=1}^m (\bh_t^\top\tilde{\bA}_i\bh_t) (\bh_{t+1}^\top \tilde{\bA}_i\bx_t)\Big| + \label{eq:main_control_gen}\\
        &4\mu \cdot \Big|\frac{1}{m}\sum_{i=1}^m (\bh_t^\top\tilde{\bA}_i\bx) (\bh_{t+1}^\top \tilde{\bA}_i\bh_t)\Big|  \nn   -  \frac{4\mu}{m}\sum_{i=1}^m (\bh_t^\top\tilde{\bA}_i\bx) (\bh_{t+1}^\top \tilde{\bA}_i\bx),
    \end{align}
    where~\eqref{eq:main_control_gen_0} follows from $\bx_t^\top\tilde{\bA}_i\bx_t-y_i = (\bx+\bh_t)^\top\tilde{\bA}_i(\bx+\bh_t)-\bx^\top\tilde{\bA}_i\bx = \bh_t^\top\tilde{\bA}_i\bh_t + 2\bh_t^\top \tilde{\bA}_i\bx$. For $\delta > 0$, from~\cite[Lemma~5.2]{vershynin2010introduction}, we know that there exists a $(\delta/L)$-net $M$ of $\mathbb{B}^k(r)$ such that 
    \begin{equation}\nonumber
        \log |M| \le k \log\frac{4Lr}{\delta}. 
    \end{equation}
    Since $G$ is $L$-Lipschitz, we obtain that $G(M)$ is a $\delta$-net of $\calR(G)=G(\mathbb{B}^k(r))$. Then there exist $\bu \in G(M)$ such that $\|\bx_t - \bu\|_2 \le \delta$ and $\bv \in G(M)$ such that $\|\bx_{t+1}-\bv\|_2 \le \delta$. Moreover, using \cite[Thm. 4.4.5]{vershynin2018high} and a union bound over $i\in [m]$, we obtain that with probability $1-m\exp(-\Omega(n))$ that 
    \begin{equation}\label{eq:event_A_i_bds}
     \|\tilde{\bA}_i\|_{op} \le C\sqrt{n},~~\forall~i\in[m]
    \end{equation}
    Throughout the following, we will condition on the event in~\eqref{eq:event_A_i_bds}. Then, we control the three terms (apart from the simple $\bm{h}_t^\top \bm{h}_{t+1}$) in~\eqref{eq:main_control_gen}.

    \noindent{\textbf{(Bounding the first two terms)}} We have for $i \in [m]$ that
    \begin{align}
        & (\bh_t^\top\tilde{\bA}_i\bh_t) (\bh_{t+1}^\top\tilde{\bA}_i \bx_t)  \nn \\
        &=  ((\bu-\bx)^\top\tilde{\bA}_i(\bu-\bx)+2(\bx_t - \bu)^\top \tilde{\bA}_i(\bu-\bx)\\& \quad +\nn (\bx_t-\bu)^\top\tilde{\bA}_i(\bx_t-\bu)) (\bh_{t+1}^\top\tilde{\bA}_i \bx_t).\label{eq:tedious_eq1_1}
    \end{align}
    We have 
    \begin{align}
        \nn&\left|2((\bx_t - \bu)^\top \tilde{\bA}_i(\bu-\bx))(\bh_{t+1}^\top\tilde{\bA}_i \bx_t)\right|\\ \le &2C\delta n \|\bu-\bx\|_2 \|\bh_{t+1}\|_2 \\
        \le &2C\delta n(\|\bh_t\|_2+\delta) \|\bh_{t+1}\|_2,\label{eq:tedious_eq1_2}
    \end{align}
    and 
    \begin{align}
        \Big|((\bx_t-\bu)^\top\tilde{\bA}_i(\bx_t-\bu))(\bh_{t+1}^\top\tilde{\bA}_i \bx_t)\Big| \le C\delta^2 n\|\bh_{t+1}\|_2. \label{eq:tedious_eq1_3}
    \end{align}
    In addition, we have
    \begin{align}
        & \nn((\bu-\bx)^\top\tilde{\bA}_i(\bu-\bx))(\bh_{t+1}^\top\tilde{\bA}_i \bx_t) \\&= (\bu-\bx)^\top\tilde{\bA}_i(\bu-\bx))((\bv -\bx)^\top\tilde{\bA}_i \bu) \label{eq:tedious_eq1_4} \\
        & + (\bu-\bx)^\top\tilde{\bA}_i(\bu-\bx))((\bx_{t+1} -\bv)^\top\tilde{\bA}_i \bu + \bh_{t+1}^\top\tilde{\bA}_i (\bx_t-\bu)),\nn
    \end{align}
    with 
    \begin{align}
        \nn&\left|(\bu-\bx)^\top\tilde{\bA}_i(\bu-\bx))((\bx_{t+1} -\bv)^\top\tilde{\bA}_i \bu)\right| \\&\quad\le Cn\delta (\|\bh_t\|_2 + \delta)^2 \label{eq:tedious_eq1_5}
    \end{align}
    and
    \begin{align}
       \nn & \left|(\bu-\bx)^\top\tilde{\bA}_i(\bu-\bx))(\bh_{t+1}^\top\tilde{\bA}_i (\bx_t-\bu))\right| \\&\quad\le Cn\delta (\|\bh_t\|_2 + \delta)^2 \|\bh_{t+1}\|_2. \label{eq:tedious_eq1_6}
    \end{align}
    Using Lemma 14 with taking a union bound over $(G(M)-\bx) \times (G(M)-\bx) \times (G(M)-\bx) \times G(M)$ and setting $\epsilon = C\sqrt{\frac{k \log \frac{4Lr}{\delta}}{m}}$, we obtain with probability $1-2\exp\big(-\Omega(k \log \frac{4Lr}{\delta})\big)$ that 
    \begin{align}
    \begin{split}
        & \Big|\frac{1}{m}\sum_{i=1}^m (\bu-\bx)^\top\tilde{\bA}_i(\bu-\bx))((\bv -\bx)^\top\tilde{\bA}_i \bu)\nn\\&\quad\quad - ((\bu-\bx)^\top(\bv -\bx))((\bu-\bx)^\top\bu)\Big| \nn \\
        & \le C\sqrt{\frac{k \log \frac{4Lr}{\delta}}{m}} \cdot\|\bu-\bx\|_2^2 \cdot\|\bv -\bx\|_2 \\
        & \le C\sqrt{\frac{k \log \frac{4Lr}{\delta}}{m}}\cdot (\|\bh_t\|_2+\delta)^2 \cdot(\|\bh_{t+1}\|_2+\delta).
    \end{split}
    \end{align}
    Then, when setting $m = \Omega\big(k \log \frac{Lr}{\delta}\big)$ with a large enough implied constant, we obtain the following where $c$ can be made sufficiently small:
    \begin{align}
        \nn &\Big|\frac{1}{m}\sum_{i=1}^m (\bu-\bx)^\top\tilde{\bA}_i(\bu-\bx))((\bv -\bx)^\top\tilde{\bA}_i \bu)\Big|\\&\quad\le (1+c)(\|\bh_t\|_2+\delta)^2 \cdot(\|\bh_{t+1}\|_2+\delta). \label{eq:tedious_eq1_7}
    \end{align}
Combining~\eqref{eq:tedious_eq1_1},~\eqref{eq:tedious_eq1_2},~\eqref{eq:tedious_eq1_3},~\eqref{eq:tedious_eq1_4},~\eqref{eq:tedious_eq1_5},~\eqref{eq:tedious_eq1_6},~\eqref{eq:tedious_eq1_7}, we obtain that when $m = \Omega\big(k \log \frac{Lr}{\delta}\big)$, with probability $1-\exp(-\Omega(k\log\frac{Lr}{\delta}))$, 
    \begin{align}
        \begin{split}\nonumber
       & \Big|\frac{1}{m}\sum_{i=1}^m (\bh_t^\top\tilde{\bA}_i\bh_t) (\bh_{t+1}^\top\tilde{\bA}_i \bx_t) \Big| \nn \\
       & \le \Big(\|\bh_t\|_2+\frac{3}{2}\delta\Big) \Big\{\Big[(1+c)(\|\bh_{t+1}\|_2+\delta) + {3Cn\delta}\Big]\\&\quad\quad\quad\quad\quad \nn \cdot(\|\bh_t\|_2+\delta) +  2Cn\delta\|\bh_{t+1}\|_2\Big\}
       \end{split}\\
        & \le \Big(\|\bh_t\|_2+\frac{3}{2}\delta\Big) \Big\{(1+c)\|\bh_t\|_2 \cdot \|\bh_{t+1}\|_2 \nn\\& \quad\quad\quad\quad\quad+  4Cn\delta (\|\bh_t\|_2+\|\bh_{t+1}\|_2+\delta)\Big\}. \label{eq:tedious_eq1_final}
    \end{align}
    In addition, running the mechanism for proving (\ref{eq:tedious_eq1_final}) again, we find that the same bound remains valid for the second term:
    \begin{align}
        &\Big|\frac{1}{m}\sum_{i=1}^m (\bh_t^\top \tilde{\bA}_i \bx)  (\bh_{t+1}^\top \tilde{\bA}_i \bh_t) \Big|   \le \Big(\|\bh_t\|_2+\frac{3}{2}\delta\Big)\label{eq:tedious_eq2_final}\\
        &\Big((1+c)\|\bh_t\|_2 \cdot \|\bh_{t+1}\|_2 + {4Cn\delta}(\|\bh_t\|_2+\|\bh_{t+1}\|_2+\delta)\Big). \nn
    \end{align}
    \textbf{(Bounding the third term} $\frac{1}{m}\sum_{i=1}^m (\bh_t^\top \tilde{\bA}_i \bx)  (\bh_{t+1}^\top \tilde{\bA}_i \bx)$) We have 
    \begin{align}
        & (\bh_t^\top \tilde{\bA}_i \bx)  (\bh_{t+1}^\top \tilde{\bA}_i \bx) = ((\bu-\bx)^\top \tilde{\bA}_i \bx)  ((\bv-\bx)^\top \tilde{\bA}_i \bx) \nn\\
        & \quad + ((\bx_t-\bu)^\top \tilde{\bA}_i \bx)  (\bh_{t+1}^\top \tilde{\bA}_i \bx)\label{eq:tedious_eq3_1} \\&\quad + ((\bu-\bx)^\top \tilde{\bA}_i \bx)  ((\bx_{t+1}-\bv)^\top \tilde{\bA}_i \bx) \nn 
    \end{align}
    with 
    \begin{equation}
        \big|((\bx_t-\bu)^\top \tilde{\bA}_i \bx)  (\bh_{t+1}^\top \tilde{\bA}_i \bx)\big| \le Cn\delta\|\bh_{t+1}\|_2\label{eq:tedious_eq3_2}
    \end{equation}
    and 
    \begin{equation}
    \big|((\bu-\bx)^\top \tilde{\bA}_i \bx)  ((\bx_{t+1}-\bv)^\top \tilde{\bA}_i \bx) \big| \le Cn\delta (\|\bh_t\|_2 + \delta). \label{eq:tedious_eq3_3}
    \end{equation}
    In addition, using Lemma~14 with taking a union bound over $(G(M)-\bx) \times \{\bx\} \times (G(M)-\bx) \times \{\bx\}$ and setting $\epsilon = C\sqrt{\frac{k \log (nLr)}{m}}$, we obtain that when $m = \Omega\big(k \log \frac{Lr}{\delta}\big)$, with probability $1-\exp(-\Omega(k\log\frac{Lr}{\delta}))$, 
    \begin{align}\nonumber
    \begin{split}
       &  \Big|\frac{1}{m}\sum_{i=1}^m ((\bu-\bx)^\top \tilde{\bA}_i \bx)  ((\bv-\bx)^\top \tilde{\bA}_i \bx)  \nn\\&\quad -\frac{1}{2}\left((\bu-\bx)^\top(\bv-\bx) + ((\bu-\bx)^\top\bx)((\bv-\bx)^\top\bx)\right)\Big|\nn\\
        & \le C\sqrt{\frac{k\log\frac{4Lr}{\delta}}{m}} \|\bu-\bx\|_2 \cdot\|\bv-\bx\|_2  \\
        & \le C\sqrt{\frac{k\log\frac{4Lr}{\delta}}{m}} (\|\bh_t\|_2+\delta) \cdot(\|\bh_{t+1}\|_2+\delta),
    \end{split}
    \end{align}
    which leads to
    \begin{align}
        \begin{split}\nonumber
        & \frac{1}{m}\sum_{i=1}^m ((\bu-\bx)^\top \tilde{\bA}_i \bx)  ((\bv-\bx)^\top \tilde{\bA}_i \bx) \nn \\
        & \nn \ge \frac{1}{2}(\bu-\bx)^\top(\bv-\bx) + \frac{1}{2}((\bu-\bx)^\top\bx)((\bv-\bx)^\top\bx) \\&\quad - C\sqrt{\frac{k\log\frac{4Lr}{\delta}}{m}} (\|\bh_t\|_2+\delta) \cdot(\|\bh_{t+1}\|_2+\delta) \\
        & \ge \frac{1}{2}(\bh_t^\top\bh_{t+1} - \delta(\|\bh_t\|_2 + \|\bh_{t+1}\|_2 + \delta))\nn 
        \end{split}\\
        &\quad- C\sqrt{\frac{k\log(nLr)}{m}} (\|\bh_t\|_2+\delta) \cdot(\|\bh_{t+1}\|_2+\delta),\label{eq:tedious_eq3_4}
    \end{align}
    where we use $(\bu-\bx)^\top(\bv-\bx) \ge \bh_t^\top\bh_{t+1} - \delta(\|\bh_t\|_2 + \|\bh_{t+1}\|_2 + \delta)$ and $(\bu-\bx)^\top\bx)((\bv-\bx)^\top\bx) \ge 0$ in~\eqref{eq:tedious_eq3_4}. Combining~\eqref{eq:tedious_eq3_1},~\eqref{eq:tedious_eq3_2},~\eqref{eq:tedious_eq3_3}, and~\eqref{eq:tedious_eq3_4}, and setting $m = \Omega\big(k \log \frac{Lr}{\delta}\big)$ with a large enough implied constant, we obtain the following where $c$ can be made sufficiently small
    \begin{align}
        \begin{split}\nonumber
        & -\frac{1}{m}\sum_{i=1}^m (\bh_t^\top \tilde{\bA}_i \bx)  (\bh_{t+1}^\top \tilde{\bA}_i \bx) \nn \\
        & \le -\frac{1}{2}\bh_t^\top\bh_{t+1} + \frac{\delta}{2}(\|\bh_t\|_2 + \|\bh_{t+1}\|_2 + \delta) \nn\\& \quad\nn+ c (\|\bh_t\|_2+\delta) \cdot(\|\bh_{t+1}\|_2+\delta) \nn \\&\quad\quad + {Cn\delta} (\|\bh_t\|_2 + \|\bh_{t+1}\|_2 + \delta)\nn 
        \end{split}\\
        & \le \Big(-\frac{1}{2}+c\Big) \|\bh_t\|_2 \cdot \|\bh_{t+1}\|_2 \nn\\&\quad \quad + {2Cn\delta} (\|\bh_t\|_2 + \|\bh_{t+1}\|_2 + \delta).  \label{eq:tedious_eq3_final}
    \end{align}
   {\textbf{(Combining everything)}} Note that $\delta < 1$ and $\|\bh_t\|_2 \le 2$. Combining~\eqref{eq:main_control_gen} and~\eqref{eq:tedious_eq1_final},~\eqref{eq:tedious_eq2_final},~\eqref{eq:tedious_eq3_final}, we obtain
    \begin{align}\nonumber
    \begin{split}
        &\|\bh_{t+1}\|_2^2  \le (2-2\mu+4c\mu)\|\bh_t\|_2\cdot\|\bh_{t+1}\|_2 \nn\\&\quad + 6\mu(1+c)\Big(\|\bh_t\|_2+\frac{3}{2}\delta\Big)\|\bh_t\|_2\cdot\|\bh_{t+1}\|_2 \nn \\
        & \quad + {24C\mu n \delta}  \Big(\frac{1}{3}+\|\bh_t\|_2 + \frac{3}{2}\delta\Big) (\|\bh_t\|_2 + \|\bh_{t+1}\|_2 + \delta) \nn\\
        & \le \Big[\big(2-2\mu+4c\mu+ 25C\mu n \delta \big)\|\bh_t\|_2\nn \\&\quad+7\mu\|\bh_t\|_2^2 +  48C\mu n \delta\Big] \|\bh_{t+1}\|_2 + {288C\mu n \delta} ,
    \end{split}
    \end{align}
    which gives
    \begin{align*}
        \|\bh_{t+1}\|_2 &\le \Big(2-2\mu+4c\mu +7\mu\|\bh_t\|_2 + 25C\mu n\delta \Big)\|\bh_t\|_2 \\&\quad+  48C\mu n \delta  +   \sqrt{288C\mu n \delta}. 
    \end{align*}
    Since $\delta = O(1/n)$, we obtain
    \begin{equation}\label{(36)}
         \|\bh_{t+1}\|_2 \le  (2-2\mu +7\mu\|\bh_t\|_2+c_1\mu)\|\bh_t\|_2 + C\sqrt{\mu n \delta} ,
    \end{equation}
    where $c_1 >0$ can be made sufficiently small. Therefore, if for some given $\varepsilon>0$ the initial vector $\bx_0$ satisfies the condition  
    $0 \le 2-2\mu +7\mu\|\bx_0 - \bx\|_2     < 1-2\varepsilon,$
    we can make $c_1$ in (\ref{(36)}) small enough so that $2-2\mu +7\mu\|\bm{h}_t\|_2 +c_1\mu<1-\varepsilon$, which leads to the following for any $t\geq 0$:$ \|\bm{h}_{t+1}\|_2 \leq (1-\varepsilon)\|\bm{h}_t\|_2+ C\sqrt{n\delta}.$ 
Moreover, by induction this yields $ \|\bm{h}_t\|_2 \leq (1-\varepsilon)^t\|\bm{h}_t\|_2 + O(\sqrt{n\delta})$ 
for any $t\geq 0$.

\noindent{\textbf{(Recaling $\delta$)}}
By setting $\delta=\frac{\delta_1^2}{n}$ and noting that $\delta = O(1/n)$ is equivalent to $\delta_1 = O(1)$, we obtain that under $m=\Omega(k\log \frac{Lrn}{\delta_1^2})$, with probability at least $1-C_1\exp(-C_2k\log \frac{Lrn}{\delta^2})$,  
     it holds for any $t\geq 0$ that $\|\bm{h}_t\|_2 \leq (1-\varepsilon)^t \|\bm{h}_t\|_2+O(\delta_1)$, as desired. \hfill $\square$

\subsection{Proofs of technical lemmas in sparse case (Lemmas 11-12)}
\subsubsection{\textbf{Proof of Lemma 11}}
We only need to prove the statement before ``In particular". Because $\mathrm{supp}(\bm{x})\subset S$,
    from Lemma 3 in the paper, we have $\mathbbm{E}[\bm{S}_{in}]_{S,S}=\mathbbm{E}[\bm{S}_{in}^\top]_{S,S}=\bm{xx}^\top$. Thus, by letting \begin{equation}\label{defi:calX}
        \mathcal{X}=\{\bm{u}\in \mathbb{R}^n:\mathrm{supp}(\bm{u})\subset [k],~\|\bm{u}\|_2=1\}
    \end{equation} there exists some $\bm{u}_0\in \mathcal{X}$ and $\bm{v}_0\in \mathcal{X}$, such that we have\begin{equation}
        \begin{aligned}\label{4.2}
            &\|[\bm{\widetilde{S}}_{in}]_{S,S}-\bm{xx}^\top\|_{op} \leq \|[\bm{S}_{in}]_{S,S}-\bm{xx}^\top\|_{op}\\&=\Big\|\frac{1}{m}\sum_{i=1}^m\Big(y_i[\bm{A}_i]_{S,S}-\mathbbm{E}\big(y_i[\bm{A}_i]_{S,S}\big)\Big)\Big\|_{op}\\
            &=\frac{1}{m}\sum_{i=1}^m\Big(y_i \bm{u}_0^\top\bm{A}_i\bm{v}_0-\mathbbm{E}\big[y_i \bm{u}_0^\top\bm{A}_i\bm{v}_0\big]\Big)\\&=\bm{u}_0^\top\big(\bm{S}_{in}-\bm{xx}^\top\big)\bm{v}_0.
        \end{aligned}
    \end{equation}
    We construct $\mathcal{N}_{1/8}$ as a $\frac{1}{8}$-net of $\mathcal{X}$, and we assume $|\mathcal{N}_{1/8}|\leq 17^k$. Note that for any fixed $(\bm{u},\bm{v})\in \mathbb{S}^{n-1}\times \mathbb{S}^{n-1}$, we have $\|y_i\bm{u}^\top\bm{A}_i\bm{v}\|_{\psi_1}\leq \|y_i\|_{\psi_2}\|\bm{u}^\top\bm{A}\bm{v}\|_{\psi_2}=O(1)$, so Bernstein's inequality yields for ant $t>0$ that, with probability at least $1-2\exp\big(-cm\min\{t,t^2\}\big)$, we have
    \begin{equation}\label{(3)}
         \frac{1}{m}\Big|\sum_{i=1}^m\Big(y_i\bm{u}^\top\bm{A}_i\bm{v}-\mathbbm{E}\big[y_i\bm{u}^\top\bm{A}_i\bm{v}\big]\Big)\Big|\leq t 
    \end{equation}
    Thus, we take a union bound over $(\bm{u},\bm{v})\in \mathcal{N}_{1/8}\times \mathcal{N}_{1/8}$ to obtain that with probability at least $1- 2\exp\big(2k\log 17-cm\min\{t,t^2\}\big)$,
    \begin{equation}\label{add1}
    \begin{aligned}
         \sup_{\bm{u}\in \mathcal{N}_{1/8}}\sup_{\bm{v}\in \mathcal{N}_{1/8}}\frac{1}{m}\Big|\sum_{i=1}^m\Big(y_i\bm{u}^\top\bm{A}_i\bm{v}-\mathbbm{E}\big[y_i\bm{u}^\top\bm{A}_i\bm{v}\big]\Big)\Big|\leq t 
        \end{aligned}
    \end{equation}
    Therefore,   setting $t=C_1\big(\sqrt{\frac{k}{m}}+\frac{k}{m}\big)$ with sufficiently large $C_1$, we obtain that the following bound with probability at least $1-2\exp(-k)$: 
    \begin{equation}\label{4.4}\begin{aligned}
        &\sup_{\bm{u}\in \mathcal{N}_{1/8}}\sup_{\bm{v}\in \mathcal{N}_{1/8}}\frac{1}{m}\Big|\sum_{i=1}^m\Big(y_i\bm{u}^\top\bm{A}_i\bm{v}-\mathbbm{E}\big[y_i\bm{u}^\top\bm{A}_i\bm{v}\big]\Big)\Big|\\&~~~~~~~~~~~~~~~~~~~~~~~~~~~~\leq C_1\Big(\sqrt{\frac{k}{m}}+\frac{k}{m}\Big).\end{aligned}
    \end{equation}
    Now, by construction of $\mathcal{N}_{1/8}$, we can pick $\bm{u}_1\in\mathcal{N}_{1/8}$ and $\bm{v}_1\in\mathcal{N}_{1/8}$ such that $\|\bm{u}_1-\bm{u}_0\|_2\leq \frac{1}{8}$ and $\|\bm{v}_1-\bm{v}_0\|_2\leq \frac{1}{8}$. Furthermore, we can bound $\bm{u}_0^\top\big(\bm{S}_{in}-\bm{xx}^\top\big)\bm{v}_0$ as follows:
    \begin{equation}
        \begin{aligned}\label{add2}
            &~~~~\bm{u}_0^\top\big(\bm{S}_{in}-\bm{xx}^\top\big)\bm{v}_0 \\&=(\bm{u}_0-\bm{u}_1)^\top\big(\bm{S}_{in}-\bm{xx}^\top\big)\bm{v}_0+\bm{u}_0^\top\big(\bm{S}_{in}-\bm{xx}^\top\big)\bm{v}_0\\&~~~~~~~~~~~+\bm{u}_1^\top\big(\bm{S}_{in}-\bm{xx}^\top\big)(\bm{v}_0-\bm{v}_1)\\
            &\leq \frac{1}{4}\big\|[\bm{S}_{in}]_{S,S}-\bm{xx}^\top\big\|_{op}\\&~~~~~~~~~~~~+\sup_{\bm{u}\in \mathcal{N}_{1/8}}\sup_{\bm{v}\in \mathcal{N}_{1/8}}\bm{u}^\top\big(\bm{S}_{in}-\bm{xx}^\top\big)\bm{v}\\
            &\leq \frac{1}{4}\bm{u}_0^\top\big(\bm{S}_{in}-\bm{xx}^\top\big)\bm{v}_0 + C_1\Big(\sqrt{\frac{k}{m}}+\frac{k}{m}\Big),
        \end{aligned}
    \end{equation}
    where the first inequality follows by observing $\bm{u}_0-\bm{u}_1=r\bm{u}_2$ for some $0\leq r\leq \frac{1}{8}$ and $\bm{u}_2\in \mathcal{X}$ (and similar relation for $\bm{v}_0-\bm{v}_1$), the second inequality is due to (\ref{4.2}) and (\ref{4.4}). Rearranging and using (\ref{4.2}) again completes the proof. 
\hfill $\square$
    \subsubsection{\textbf{Proof of Lemma 12}} 
It is not hard to see that $\mathbbm{E}\big[(\bm{p}^\top \bm{A}_i\bm{q})(\bm{u}^\top\bm{A}_i\bm{v})\big]=(\bm{p}^\top\bm{u})(\bm{q}^\top\bm{v})$, so   the left-hand side of (\textcolor{black}{33}) in the paper equals to\begin{equation}\begin{aligned}\nonumber
        &\sup_{\bm{p},\bm{q}\in \mathcal{X}}\Big\|\frac{1}{m}\sum_{i=1}^m\big(\bm{p}^\top\bm{A}_i\bm{q}\big)[\bm{A}_i]_{S,S}-\mathbbm{E}\Big[\big(\bm{p}^\top\bm{A}_i\bm{q}\big)[\bm{A}_i]_{S,S}\Big]\Big\|_{op}\\&~~~~~~~~~~~~~~~~~~~:=\sup_{\bm{p},\bm{q}\in\mathcal{X}}\big\| \bm{G}(\bm{p},\bm{q})\big\|_{op},\end{aligned}
    \end{equation}
    and we can further pick $\bm{p}_0\in \mathcal{X}~\text{and}~\bm{q}_0\in \mathcal{X}$ such that $ \sup_{\bm{p},\bm{q}\in \mathcal{X}}\|\bm{G}(\bm{p},\bm{q})\|_{op}=\|\bm{G}(\bm{p}_0,\bm{q}_0)\|_{op}.$

    Now let us inspect the proof of Lemma 11. We find that (\ref{4.2}) and the first three lines of (\ref{add2}) imply \begin{equation}\nonumber\begin{aligned}
        &\Big\|\frac{1}{m}\sum_{i=1}^m y_i[\bm{A}_i]_{S,S}-\mathbbm{E}\big[y_i[\bm{A}_i]_{S,S}\big]\Big\|_{op}\\&\leq 2\sup_{\bm{u}\in \mathcal{N}_{1/8}}\sup_{\bm{v}\in \mathcal{N}_{1/8}}\frac{1}{m}\Big|\sum_{i=1}^m\Big(y_i\bm{u}^\top\bm{A}_i\bm{v}-\mathbbm{E}\big[y_i\bm{u}^\top\bm{A}_i\bm{v}\big]\Big)\Big|.\end{aligned}
    \end{equation} 
    Then we substitute this into (\ref{add1}) to obtain \begin{equation}
    \begin{aligned}\label{4.12}
        \mathbbm{P}&\Big(\Big\|\frac{1}{m}\sum_{i=1}^my_i[\bm{A}_i]_{S,S}-\mathbbm{E}\big[y_i[\bm{A}_i]_{S,S}\big]\Big\|_{op}\geq 2t\Big)\\
        &\leq 2\exp\big(2k\log 17-cm\min\{t,t^2\}\big),~~\forall~t>0.
        \end{aligned}
    \end{equation}
    Moreover, the proof of Lemma 11 does not rely on $y_i=\bm{x}^\top\bm{A}_i\bm{x}$; rather, the proof still holds when $y_i$ is substituted with $\bm{p}^\top\bm{A}_i\bm{q}$ for any fixed pair of $(\bm{p},\bm{q})\in \mathcal{X}\times \mathcal{X}$. More specifically, substituting $y_i$ in (\ref{4.12}) with $\bm{p}^\top\bm{A}_i\bm{q}$ for some fixed $(\bm{p},\bm{q})\in \mathcal{X}\times \mathcal{X}$ we obtain the following for any $t>0,$
    \begin{equation}\nonumber
        \begin{aligned}
        &\mathbbm{P}\Big(\big\|\bm{G}(\bm{p},\bm{q})\big\|_{op}\geq 2t\Big) \\&~~~~~~~~~\leq  2\exp\big(2k\log 17-cm\min\{t,t^2\}\big).
        \end{aligned}
    \end{equation}
    Furthermore, let $\mathcal{N}_{1/8}$ be a $\frac{1}{8}$-net of $\mathcal{X}$ satisfying $|\mathcal{N}_{1/8}|\leq 17^k$, we take a union bound over $(\bm{p},\bm{q})\in \mathcal{N}_{1/8}\times \mathcal{N}_{1/8}$, which yields for any $t>0$ that \begin{equation}\nonumber
        \begin{aligned}\label{4.14}&\mathbbm{P}\Big(\sup_{\bm{p},\bm{q}\in\mathcal{N}_{1/8}}\big\|\bm{G}(\bm{p},\bm{q})\big\|_{op}\geq 2t\Big) \\&~~~~~~~~\leq 2\exp\big(4k\log 16-cm\min\{t,t^2\}\big).\end{aligned}
    \end{equation}
    With this in place, we are now ready to bound $\|\bm{G}(\bm{p}_0,\bm{q}_0)\|_{op}$. We pick $\bm{p}_1\in \mathcal{N}_{1/8}$ and $\bm{q}_1\in \mathcal{N}_{1/8}$ so that $\|\bm{p}_1-\bm{p}_0\|_{2}\leq \frac{1}{8}$ and $\|\bm{q}_1-\bm{q}_0\|_2\leq \frac{1}{8}$ hold. Then, by the bilinearity of $\bm{G}(\bm{p},\bm{q})$ on $(\bm{p},\bm{q})$ we have 
    \begin{equation}\nonumber
        \begin{aligned}
         &   \big\|\bm{G}(\bm{p}_0,\bm{q}_0)\big\|_{op}\\&=\big\|\bm{G}(\bm{p}_0-\bm{p}_1,\bm{q}_0)+\bm{G}(\bm{p}_1,\bm{q}_0-\bm{q}_1)+\bm{G}(\bm{p}_1,\bm{q}_1)\big\|_{op}\\
            &\leq \frac{1}{4}\big\|\bm{G}(\bm{p}_0,\bm{q}_0)\big\|_{op}+ \sup_{\bm{p},\bm{q}\in \mathcal{N}_{1/8}}\big\|\bm{G}(\bm{p},\bm{q})\big\|_{op},
        \end{aligned}
    \end{equation}
    which yields $$\|\bm{G}(\bm{p}_0,\bm{q}_0)\|_{op}\leq 2\sup_{\bm{p},\bm{q}\in \mathcal{N}_{1/8}}\|\bm{G}(\bm{p},\bm{q})\|_{op}.$$ Recall that   $\sup_{\bm{p},\bm{q}\in \mathcal{X}}\|\bm{G}({\bm{p},\bm{q}})\|_{op}=\|\bm{G}(\bm{p}_0,\bm{q}_0)\|_{op}$, combined with (\ref{4.14}) we obtain for any $t>0$ that, 
    \begin{equation}\nonumber
        \begin{aligned}
        &\mathbbm{P}\Big(\sup_{\bm{p},\bm{q}\in\mathcal{X}}\big\|\bm{G}(\bm{p},\bm{q})\big\|_{op}\geq 4t\Big)\\&~~~~~~~~~\leq 2\exp\big(4k\log 16-cm\min\{t,t^2\}\big).
        \end{aligned}
    \end{equation}
    Then we set $t=C\sqrt{\frac{k}{m}}$ with large enough $C$, and  note that $t^2\lesssim t$ due to $m\gtrsim k$, 
    the desired (\textcolor{black}{33}) (in the paper) follows.
    Some simple algebra finds that (\textcolor{black}{33}) (in the paper) implies (\textcolor{black}{34}) (in the paper). The proof is complete. \hfill $\square$ 

    \subsection{Proofs of technical lemmas in generative case (Lemmas 13-14)}

\subsubsection{\textbf{Proof of Lemma 13}} We first establish (\ref{lem131}).
    Note that (\ref{(3)}) in the proof of Lemma 11 implies that, for  fixed $\bm{s}_1,\bm{s}_2\in \mathbb{R}^n$ and any $t>0$, with probability at least $1-2\exp(-cm\min\{t,t^2\})$ we have 
    \begin{equation}
        \label{(13)}\big|\bm{s}_1^\top(\bm{\widetilde{S}}_{in}-\bm{xx}^\top)\bm{s}_2\big|\lesssim 2t \cdot\|\bm{s}_1\|_2\|\bm{s}_2\|_2.
    \end{equation}
    Then a union bound   gives that, (\ref{(13)}) holds for all $(\bm{s}_1,\bm{s}_2)\in\mathcal{S}_1\times \mathcal{S}_2$ with probability at least $1-2|\mathcal{S}_1||\mathcal{S}_2|\exp(-cm\min\{t,t^2\})$. Now we set $t= C\big(\frac{\log(|\mathcal{S}_1|\cdot|\mathcal{S}_2|)}{m}+\sqrt{\frac{\log(|\mathcal{S}_1|\cdot|\mathcal{S}_2|)}{m}}\big)$ with large enough $C$, and because $m\gtrsim \log(|\mathcal{S}_1|\cdot|\mathcal{S}_2|)$, we obtain that with probability at least $1-2\exp(-\Omega(\log|\mathcal{S}_1|+\log|\mathcal{S}_2|))$,  (\ref{(13)}) with $t\leq C_1\sqrt{\frac{\log(|\mathcal{S}_1|\cdot|\mathcal{S}_2|)}{m}}$ holds for all $(\bm{s}_1,\bm{s}_2)\in\mathcal{S}_1\times \mathcal{S}_2$, as desired.

    Then, we show ``$\|\bm{\widetilde{S}}_{in}-\bm{xx}^\top\|_{op}\lesssim \frac{n}{m}$'' with high probability. Fortunately, this can be implied by Lemma 11. Specifically, we set $k=n$ in Lemma 11, then $S=[n]$, and hence it gives that with probability at least $1-2\exp(-n)$ we have $\|\bm{\widetilde{S}}_{in}-\bm{xx}^\top\|_{op}\leq \sqrt{\frac{n}{m}}+\frac{n}{m}$. The result   follows since     $m<n$. \hfill $\square$ 
    
    \subsubsection{\textbf{Proof of Lemma 14}} This is a simple outcome of Bernstein's inequality, which we apply similarly to (\ref{(3)}) in the proof of Lemma 11. Suppose that $\bm{p},\bm{q},\bm{u},\bm{v}$ are non-zero with no loss of generality. Then, because $\|\bm{p}^\top\bm{\widetilde{A}}_i\bm{q}\|_{\psi_2}=O(\|\bm{p}\|_2\|\bm{q}\|_2)$, $\|\bm{u}^\top\bm{\widetilde{A}}_i\bm{v}\|_{\psi_2}=O(\|\bm{u}\|_2\|\bm{v}\|_2)$, (\ref{(3)}) remains valid if $y_i=\bm{x}^\top\bm{A}_i\bm{x}$ (therein) is replaced with $\bm{p}^\top\bm{\widetilde{A}}_i\bm{q}/(\|\bm{p}\|_2\|\bm{q}\|_2)$ in the current setting, $\bm{u}^\top\bm{A}_i\bm{v}$ (therein) is replaced with $\bm{u}^\top\bm{\widetilde{A}}_i\bm{v}/(\|\bm{u}\|_2\|\bm{v}\|_2)$ in the current setting, which yields  the desired claim with probability at least $1-2\exp(-cm\min\{t,t^2\})$. The proof is complete by noting the additional assumption of $t\in (0,1)$. \hfill $\square$

\end{multicols*}

\end{document}

%% file: preamble.tex
\usepackage[mathscr]{eucal}
\usepackage{epsfig,epsf,psfrag}
\usepackage{amssymb,amsmath,amsfonts,latexsym}
\usepackage{amsmath,graphicx,bm,xcolor,url}
\usepackage[caption=false]{subfig} 
\usepackage{fixltx2e}
\usepackage{array}
\usepackage{verbatim}
\usepackage{bm}
\usepackage{verbatim}
\usepackage{textcomp}
\usepackage{mathrsfs}
\usepackage{epstopdf}
\usepackage{relsize}
\usepackage{cleveref} 
\usepackage{subfig}
 \usepackage{amsthm}

\catcode`~=11 \def\UrlSpecials{\do\~{\kern -.15em\lower .7ex\hbox{~}\kern .04em}} \catcode`~=13 

\allowdisplaybreaks[3]

\newcommand{\nn}{\nonumber}

\newcommand{\calP}{\mathcal{P}}

\newcommand{\calR}{\mathcal{R}}

\newcommand{\bA}{\bm{A}}

\newcommand{\bh}{\bm{h}}

\newcommand{\bu}{\bm{u}}

\newcommand{\bv}{\bm{v}}

\newcommand{\bw}{\bm{w}}

\newcommand{\bx}{\bm{x}}

\DeclareMathAlphabet{\mathbsf}{OT1}{cmss}{bx}{n}
\DeclareMathAlphabet{\mathssf}{OT1}{cmss}{m}{sl}

\DeclareSymbolFont{bsfletters}{OT1}{cmss}{bx}{n}  
\DeclareSymbolFont{ssfletters}{OT1}{cmss}{m}{n}
\DeclareMathSymbol{\bsfGamma}{0}{bsfletters}{'000}
\DeclareMathSymbol{\ssfGamma}{0}{ssfletters}{'000}
\DeclareMathSymbol{\bsfDelta}{0}{bsfletters}{'001}
\DeclareMathSymbol{\ssfDelta}{0}{ssfletters}{'001}
\DeclareMathSymbol{\bsfTheta}{0}{bsfletters}{'002}
\DeclareMathSymbol{\ssfTheta}{0}{ssfletters}{'002}
\DeclareMathSymbol{\bsfLambda}{0}{bsfletters}{'003}
\DeclareMathSymbol{\ssfLambda}{0}{ssfletters}{'003}
\DeclareMathSymbol{\bsfXi}{0}{bsfletters}{'004}
\DeclareMathSymbol{\ssfXi}{0}{ssfletters}{'004}
\DeclareMathSymbol{\bsfPi}{0}{bsfletters}{'005}
\DeclareMathSymbol{\ssfPi}{0}{ssfletters}{'005}
\DeclareMathSymbol{\bsfSigma}{0}{bsfletters}{'006}
\DeclareMathSymbol{\ssfSigma}{0}{ssfletters}{'006}
\DeclareMathSymbol{\bsfUpsilon}{0}{bsfletters}{'007}
\DeclareMathSymbol{\ssfUpsilon}{0}{ssfletters}{'007}
\DeclareMathSymbol{\bsfPhi}{0}{bsfletters}{'010}
\DeclareMathSymbol{\ssfPhi}{0}{ssfletters}{'010}
\DeclareMathSymbol{\bsfPsi}{0}{bsfletters}{'011}
\DeclareMathSymbol{\ssfPsi}{0}{ssfletters}{'011}
\DeclareMathSymbol{\bsfOmega}{0}{bsfletters}{'012}
\DeclareMathSymbol{\ssfOmega}{0}{ssfletters}{'012}

\DeclareMathOperator{\supp}{supp}

\newcommand{\qednew}{\nobreak \ifvmode \relax \else
      \ifdim\lastskip<1.5em \hskip-\lastskip
      \hskip1.5em plus0em minus0.5em \fi \nobreak
      \vrule height0.75em width0.5em depth0.25em\fi}

